\documentclass[11pt]{article}
\input{preamble}

\declareperson{Nima}
\declareperson{Kuikui}
\declareperson{Shayan}
\declareperson{Cynthia}
\declareperson{June}

\declarenameddelims{D}{{\mathcal{D}_{\mathrm{KL}}}}\lparen\rparen
\newcommand*{\B}{\mathcal{B}}
\newcommand*{\tmix}{t_{\mathrm{mix}}}
\newcommand*{\tv}{{\mathrm{TV}}}
\newcommand*{\matroid}{\mathcal{M}}
\newcommand*{\OPT}{\text{OPT}}
\DeclareMathOperator{\argmax}{argmax}
\tikzset{sptrees/.pic={
	\begin{scope}[every node/.style={circle, fill=Black, inner sep=3}]
		\node (A) at (0, 0) {};
		\node (B) at (30:1) {};
		\node (C) at (90:1.2) {};
		\node (D) at (150:0.7) {};
		\node (E) at (230:1.5) {};
		\node (F) at (300:1) {};
		\node (G) at (350:1.2) {};
	\end{scope}
	\foreach \u/\v in {A/B, B/C, C/D, D/E, E/F, F/G, G/A, A/E, A/D, A/C, E/G}
		\draw[line width=1, Gray] (\u) -- (\v);
	\foreach \u/\v in {#1}
		\draw[line width=2.5, Orange] (\u) -- (\v);
	
}}

\title{Log-Concave Polynomials IV: Approximate Exchange, Tight Mixing Times, and Near-Optimal Sampling of Forests}

\author[1]{Nima Anari}
\affil[1]{\small Stanford University, \textsf{anari@cs.stanford.edu}, \textsf{tdvuong@stanford.edu}}

\author{Kuikui Liu}
\author{Shayan Oveis Gharan}
\affil{\small University of Washington, \textsf{liukui17@cs.washington.edu}, \textsf{shayan@cs.washington.edu}}

\author{Cynthia Vinzant}
\affil{\small North Carolina State University, \textsf{clvinzan@ncsu.edu}}
\author[1]{Thuy-Duong Vuong}

\bibliography{refs}

\begin{document}
	\maketitle
	
	\begin{abstract}
		We prove tight mixing time bounds for natural random walks on bases of matroids, determinantal distributions, and more generally distributions associated with log-concave polynomials. For a matroid of rank $k$ on a ground set of $n$ elements, or more generally distributions associated with log-concave polynomials of homogeneous degree $k$ on $n$ variables, we show that the down-up random walk, started from an arbitrary point in the support, mixes in time $O(k\log k)$. Our bound has no dependence on $n$ or the starting point, unlike the previous analyses \cite{ALOV19,CGM19}, and is tight up to constant factors. The main new ingredient is a property we call approximate exchange, a generalization of well-studied exchange properties for matroids and valuated matroids, which may be of independent interest. In particular, given function $\mu: \binom{[n]}{k} \to \R_{\geq 0},$ our approximate exchange property implies that a simple local search algorithm gives a $k^{O(k)}$-approximation of $\max_{S} \mu(S)$ when $\mu$ is generated by a log-concave polynomial, and that greedy gives the same approximation ratio when $\mu$ is strongly Rayleigh. 
		
		As an application, we show how to leverage down-up random walks to approximately sample random forests or random spanning trees in a graph with $n$ edges in time $O(n\log^2 n).$ The best known result for sampling random forest was a FPAUS with high polynomial runtime recently found by \cite{ALOV19, CGM19}. For spanning tree, we improve on the almost-linear time algorithm by \textcite{Sch18}. Our analysis works on weighted graphs too, and is the first to achieve nearly-linear running time for these problems. Our algorithms can be naturally extended to support approximately sampling from random forests of size between $k_1$ and $k_2$ in time $O(n \log^2 n)$, for fixed parameters $k_1, k_2,$ as well as approximate sampling random independent set of matroid $M$ of rank $k$ on a ground set of $n$ elements using $O(kn \log k)$ calls to the independence oracle of $M$.
	\end{abstract}
	
	\section{Introduction}
\label{sec:intro}

Let $\mu:\binom{[n]}{k}\to\R_{\geq 0}$ be a density function on size $k$ subsets of $[n]=\set{1,\dots,n}$, defining a distribution $\Pr{S}\propto \mu(S)$. The generating polynomial of $\mu$ is the multivariate $k$-homogeneous polynomial defined as follows:
\[ g_\mu(z_1,\dots,z_n)=\sum_{S\in \binom{[n]}{k}} \mu(S)\prod_{i\in S} z_i. \]
We say that $g_\mu$ is log-concave if $\log(g_\mu)$ is a concave function over $\R_{\geq 0}^n$. The study of log-concave polynomials has recently enabled breakthroughs on old conjectures about matroids, including the resolution of a conjecture of \textcite{MV89} on the expansion of the bases-exchange graphs \cite{ALOV19}, and Mason's ultra-log-concavity conjecture \cite{ALOV18, BH18}. These results rely on the log-concavity of the generating polynomial for various distributions associated with matroids, most importantly the uniform distribution on the set of bases \cite{AOV18}.

Besides distributions associated with matroids, several other classes of distributions possess a log-concave generating polynomial. An important subclass consists of strongly Rayleigh distributions \cite{BBL09} which includes determinantal point processes, distributions that have found numerous applications in machine learning \cite[see][for a survey]{KT12}. A well-studied example belonging to all classes mentioned so far consists of the uniform distribution over spanning trees of a graph $G=(V, E)$. Here $n$ is the number of edges $\card{E}$ in the graph and $k$ is the number of edges in a spanning tree, i.e., $\card{V}-1$. Spanning trees of a graph form bases of a matroid called the graphic matroid \cite[see, e.g.,][]{Oxl06} and they can also be viewed as a determinantal point process because of the matrix-tree theorem \cite[see, e.g.,][]{BBL09}, and are consequently strongly Rayleigh.

The motivation behind the conjecture of \textcite{MV89} was to solve the problem of approximately sampling from bases of a matroid. After this conjecture was made, efficient sampling algorithms were developed for various special classes of matroids \cite{FM92, Gam99, JS02, JSTV04, Jer06, Clo10, CTY15, GJ18} until \textcite{ALOV19} showed an efficient approximate sampling algorithm for all matroids. This algorithm used a variant of random walks on the so-called ``bases-exchange'' graphs of matroids, that is known as the ``down-up'' random walk studied in the context of high-dimensional expanders \cite{KM16, DK17, KO20}. For a distribution defined by $\mu:\binom{[n]}{k}\to \R_{\geq 0}$, the down-up random walk $P$ starts from a set $S_0\in \binom{[n]}{k}$ and produces the Markovian sequence $S_0, S_1,S_2,\dots$ as follows:
\begin{algorithm*}
	\For{$t=0,1,2,\dots$}{
		Let $T_t\in \binom{S_t}{k-1}$ be a subset of $S_t$ obtained by dropping one element of $S_t$ uniformly at random.\;
		Let $S_{t+1}=T_t\cup\set{e}$, where the element $e$ is chosen with probability $\propto \mu(T_t\cup\set{e})$.\;
	}
\end{algorithm*}
The random walk $P$ has $\mu$ as its stationary distribution and can be efficiently implemented by probing $\mu$ on at most $n$ different sets each time. Thus assuming oracle access to $\mu$, or in the case of matroids, an independence oracle for the matroid, each step of $P$ takes $O(n)$ time. The challenging part has been establishing the mixing time of $P$, i.e., bounds on the time $t$ such that the distribution of $S_t$ is $\epsilon$-close in total variation distance to the one defined by $\mu$:
\[ \tmix(P, S_0, \epsilon):=\min\set*{t\given \norm{P^t(S_0, \cdot)-\mu(\cdot)}_\tv\leq \epsilon}. \]

\Textcite{ALOV19} proved that when $\mu$ has a log-concave generating polynomial, the spectral gap of the random walk $P$ is at least $1/k$. This implied that
\[ \tmix(P, S_0, \epsilon)\leq O\parens*{k\cdot \parens*{\log\frac{1}{\PrX{\mu}{S_0}}+\log\frac{1}{\epsilon}}}. \]
Later, \textcite{CGM19} proved a Modified Log-Sobolev Inequality (MLSI) for the same random walk which resulted in a tighter mixing time:
\[ \tmix(P, S_0, \epsilon)\leq O\parens*{k\cdot \parens*{\log\log\frac{1}{\PrX{\mu}{S_0}}+\log\frac{1}{\epsilon}}}. \]
These results lead to efficient algorithms assuming that the mass of the starting set, $\P_{\mu}{S_0}$, is not terribly small; this can often be achieved in practice. For example, for matroids, any starting basis $S_0$ will satisfy $\P_\mu{S_0}\geq 1/\binom{n}{k}\geq n^{-k}$, because the number of bases is at most $\binom{n}{k}$. Consequently the above bounds turn into $\tmix(P, S_0, \epsilon)\leq O(k(k\log(n)+\log(1/\epsilon)))$ and $\tmix(P, S_0,\epsilon)\leq O(k(\log k+\log \log n+\log (1/\epsilon)))$ respectively. However, for other distributions $\mu$ with a log-concave generating polynomial, even in the very special case of determinantal point processes, there is no control on $\min\set{\P_{\mu}{S_0}\given S_0\in \supp(\mu)}$, so one has to rely on clever tricks to find a good starting set $S_0$; even then, the best hope is to find a set $S_0$ with $\P_\mu{S_0}\gtrsim 1/\binom{n}{k}$, which results in a mixing time mildly depending on $n$.

Historically, earlier works on a subclass of matroids, called balanced matroids, followed a similar development, where initially a spectral gap result was proved, resulting in a running time\footnote{Note that the running time is $n$ times the mixing time for the down-up walk.} of $O(nk(k\log n+\log(1/\epsilon)))$ followed by MLSI which resulted in a mixing time of $O(k(\log k+\log \log n+\log (1/\epsilon)))$ \cite[see][for a survey]{MT06}. Noting that the term $\log \log n$ seems unnecessary, \textcite{MT06} raised the question of proving a better inequality that would result in a running time of $O(nk\log(k/\epsilon))$. They specifically hoped for the possibility of proving a Nash inequality, an advanced type of functional inequality used to derive very tight mixing times for some Markov chains \cite{MT06}. We believe there are barriers to using functional inequalities in general to prove $O(k\log (k/\epsilon))$ mixing time for the down-up random walk; we defer an explanation of this to a future version of this paper. However, without proving new functional inequalities, we manage to sidestep this barrier and improve the running time to the conjectured $O(nk\log (k/\epsilon))$ for not just balanced matroids, but the class of all matroids.


Our first result is a tight analysis of the mixing time, entirely removing the dependence on $\P_{\mu}{S_0}$ and $n$. 
\begin{theorem}\label{thm:main}
	For any distribution defined by $\mu:\binom{[n]}{k}\to\R_{\geq 0}$ with a log-concave generating polynomial $g_\mu$, the mixing time of the down-up random walk $P$, starting from any $S_0$ in the support of $\mu$ is
	\[ \tmix(P, S_0, \epsilon)\leq O(k\log(k/\epsilon)). \]
\end{theorem}
Note that generally we cannot hope for a better mixing time than $k\log k$; each step of the random walk $P$ replaces one element of the current set, and by a coupon collector argument, at least $\simeq k\log k$ steps are needed to replace every element of the starting set $S_0$. As long as $k$ is not too close to $n$, say $k<0.99n$, replacing every starting element is needed for sufficient mixing, even for the simple distribution $\mu$ which is uniform over $\binom{[n]}{k}$.

Our mixing time bound is an asymptotic improvement over prior work for $k=O(1)$, or more generally when $k$ is smaller than $\log(n)^\epsilon$ for all $\epsilon>0$. Another consequence of the new mixing time bound is that it enables the analysis of the down-up random walk when $n$ is infinitely large; for example, this is the case for continuous determinantal point processes \cite[see, e.g.,][]{OR18}.\footnote{We note however that one still needs to be able to implement each step of the random walk efficiently when $n$ is infinitely large. For examples where this is possible see \cite{OR18}.} To avoid complicating the notation, we do not consider infinitely large ground sets in this paper, but note that the results do generalize to such cases.


Our next result is the first quasi-linear time algorithm to sample from the uniform distribution over forests of a graph $G = G(V,E).$ This improves upon the recent result by \Textcite{ALOV19} which gives a polynomial time algorithm to sample random forest, but the run-time of this algorithm is far from being linear in the number of edges. Their algorithm samples and counts forests of fixed-size $k$ for each $k\leq \abs{V}-1$, thus takes at least $\Omega(\abs{V} \abs{E})$ time. Moreover, they employ the approximate sampling to approximate counting reduction \cite{JVV86,Anari2020IsotropyAL}, which introduces large polynomial blow-up in run-time. For application of sampling random forests, see e.g. \cite{Goel_connectivityin}.

In addition, we show a similar algorithm that also runs in quasi-linear and samples from the uniform distribution over spanning trees of $G.$ Much attention has been paid to the problem of sampling random spanning tree over the years, starting from the seminal works of \textcite{Aldous90,Broder89} who proposed a simple routine to extract a random spanning tree from the trace of a random walk on $G$ itself. Subsequent works introduced improved algorithms \cite{Wil96, CMN96, KM09, MST14, DKPRS17, DPPR17} until finally \textcite{Sch18} managed to obtain an almost-linear time algorithm running in time $n^{1+o(1)}$ on graphs with $n$ edges. This algorithm and that of several prior works were all based on the original work of \textcite{Aldous90, Broder89}; they achieved an improved running time by employing several clever, but complicated, tricks to shortcut the trace of a random walk over $G$. Our algorithms to sample a random spanning tree or random forest is wholly different, based on the down-up random walk, that achieves a nearly-linear running time of $n\log^2(n)$, while being arguably much simpler to describe and implement. Our algorithms can be naturally extended to sample from \textit{weighted} distribution over forests or spanning trees.
\begin{theorem} \label{thm:forests}
	There is an algorithm that takes a weighted graph $G = G (V, E)$ on $n$ edges with weight function $w: E \to \R_{\geq 0},$ parameters $q
	\geq 0$ and $\epsilon>0$ as input and outputs a forest $F \subseteq E$ in time $O(n\log(n)\log(n/\epsilon))$; the distribution of $F$ is guaranteed to be $\epsilon$-close in total variation distance to the distribution $\mu$ over forests of $G$ defined by 
	$\mu(F) \propto  q^{k-\abs{F}} w^F$ where $k$ is the rank of the graphic matroid of $G$, and $\abs{F}$ denotes the number of edges in $F.$
	
	In particular, for $w(e) = 1 \forall e\in E$, $\mu$ is the uniform distribution on forests of $G$ if $q=1,$ and is the uniform distribution on spanning trees of $G$ if $G$ is connected and $q=0$. 
\end{theorem}
In fact, we can extend \cref{thm:forests} to allow sampling from the uniform distribution over forests of size between $k_1$ and $k_2$, for any parameters $k_1, k_2,$ in quasi-linear time. 
\begin{theorem} \label{thm:constrainedForests}
There are algorithms that takes a weighted graph $G = G (V, E)$ on $n$ edges with weight function $w: E \to \R_{\geq 0},$ parameters $q
	\geq 0$, $k_1, k_2\in \N$ and $\epsilon>0$ as input and outputs a forest $F \subseteq E$ in time $O(n\log(n)\log(n/\epsilon))$; the distribution of $F$ is guaranteed to be $\epsilon$-close in total variation distance to the distribution $\mu^{(k_1, k_2)}$ over forests of $G$ defined by $\mu^{(k_1, k_2)} (F) \propto q^{k_2-\abs{F}} w^F$ if $\abs{F} \in [k_1, k_2],$ and $0$ otherwise.

\end{theorem}

Since our algorithm(s) is based on the MCMC method, they can only approximately sample from the forest or spanning tree distribution. In contrast, some of the prior works, including \cite{Sch18}, can sample exactly from the spanning tree distribution. This is mostly an inconsequential difference in practice, as no polynomial-time user of the algorithm can sense a difference between exact sampling and approximate sampling; one simply needs to set $\epsilon$ to be inverse-polynomially small.

We remark that our technique also leads to algorithm(s) that perform the more general task of approximately sampling from the uniform distribution over the family of independent sets of an arbitrary matroid, given access to suitable oracles. Specifically, for a matroid $\matroid = ([n], \mathcal{I})$ of rank $k$, an algorithm similar to the one from \cref{thm:forests} samples from a distribution that is $\epsilon$-close to the uniform distribution over the family of independent sets $\mathcal{I}$ using $O(n \log \frac{n}{\epsilon})$ time calls to oracle $\mathcal{O}'$ whose input-output behavior is described by:
\begin{itemize}
    \item $\mathcal{O}'$ takes as input a set $S\subseteq [n]$ that is guaranteed to contains at most one circuit
    \item outputs a uniformly random element from the unique circuit in $S$ if, such a circuit exists
\end{itemize}
For graphic matroid, we can implement $\mathcal{O'}$ with amortized quasi-constant query time using link-cut tree. In general,
since the input $S$ is guaranteed to have size at most $k+1$, we can implement each call to $\mathcal{O'}$ using $O(k)$ calls to the more familiar \textit{independent set oracle} $\mathcal{O}_I$ for $\matroid,$ resulting in a $O(kn \log \frac{n}{\epsilon})$-time algorithm. On the other hand, a modified version of the algorithm from \cref{thm:forests} can perform the same sampling task in $O(nk \log \frac{k}{\epsilon})$ calls to $\mathcal{O}_I.$ Though these algorithms still has sub-optimal runtime, they are an improvement upon \cite{ALOV19}'s algorithm. Their algorithm involves running the down-up walk to generate uniform samples from size-$\ell$ independent sets for each $\ell \leq k,$ then counting size-$\ell$ independent sets via a sampling to counting reduction
which would introduce a large run-time blow-up.

\subsection{Techniques}

In order to prove \cref{thm:main}, our strategy is to combine a new analysis of the initial steps of the down-up random walk with the previously known Modified Log-Sobolev Inequality \cite{CGM19}. Specifically we show that conditioned on having replaced every element of the starting set $S_0$ at least once by time $t$, the set at time $t$ can be used as a \emph{warm start} for the rest of the steps. Specifically, we show that the density of the set at time $t$ w.r.t.\ $\mu$, conditioned on this event, is upper-bounded by only a function of $k$.

In order to prove this, we introduce a new property of functions $\mu:\binom{[n]}{k}\to \R_{\geq 0}$ that we call $\alpha$-approximate exchange. This property says that for every $S, T\in \binom{[n]}{k}$, and $i\in S$, there exists $j\in T$ such that
\[ \mu(S)\mu(T)\leq \alpha\cdot \mu(S-i+j)\mu(T+i-j). \]
Note that when $\mu$ takes values in $\set{0,1}$ and $\alpha\geq 1$, this property becomes equivalent to the famous strong basis exchange axiom of matroids \cite{Oxl06}; if $\B=\mu^{-1}(1)$ is the family of sets indicated by $\mu$, this property says that for every $S, T\in \B$ and $i\in S$, there exists $j\in T$ such that $S-i+j\in \B$ and $T+i-j\in \B$. This property can be seen as a quantitative variant of strong basis exchange. Alternatively, it can be viewed as an approximate and multiplicative form of $M^\natural$-concavity, a cornerstone of discrete convex analysis \cite{MS99}. We prove that every $\mu$ with a log-concave generating polynomial satisfies $2^{O(k)}$-approximate exchange. Crucially, our $\alpha$ does not depend on $n$. We remark that \textcite{BH18} showed a result that can be thought of as a converse to this. They proved that $M^\natural$-concavity of $\log \mu$, equivalent to $1$-approximate exchange property, implies that the generating polynomial of $\mu$ is log-concave. We show that a similar approximate exchange property implies that a simple local search algorithm gives a $k^{2k}$-approximation on the problem of maximizing $\mu(S)$ for $\mu :\binom{[n]}{k}\to \R_{\geq 0}$ generated by a log-concave polynomial (\cref{lem:localsearch}). If the generating polynomial of $\mu :\binom{[n]}{k}\to \R_{\geq 0}$ is moreover strongly Rayleigh, then $\mu$ satisfies a slightly stronger exchange property (see \cref{lemma:rsExchange}) that in turn implies greedy gives a $k^{2k}$-approximation of $\max_S \mu(S)$ (see \cref{lem:greedy}). This is a generalization of \Textcite{Kha95}'s classical result that greedy produces a $k^{O(k)}$-approximation of the (sub)determinant maximization problem \cite{Kha95, DEFM14, Nik15}, as well as \cite{PackerO4,CM13}'s more recent result that greedy gives $k^{O(k)}$-approximation for the largest $j$-dimensional simplex problem. The best result on the largest $j$-dimensional simplex problem is a $2^{O(k)}$-approximation by \textcite{Nik15}, matching the lower bound given by \cite{DEFM14,CM13}.   

We discuss the high-level ideas for proving \cref{thm:forests}. For simplicity's sake, we consider the unweighted case i.e. $w(e)=1 \forall e\in E$.
It would be helpful to first discuss the special case $q=0$, $G$ is connected, and $\mu$ is uniformly distributed over spanning trees of $G$. We would like to use the down-up random walk from \cref{thm:main} to sample from $\mu$. 
Though the down-up walk on the support of $\mu$ mixes in nearly-linear time, we do not see a way to implement each step of it in polylogarithmic time. Fortunately, the down-up random walk on an equivalent family of sets, the \textit{dual} of the graphic matroid of $G$, which consists of the \emph{complements} of spanning trees, also mixes fast, and we can implement each step in amortized $O(\log n)$-time using link-cut tree \cite{ST83,RTF18}.

For $q\neq 0$, the distribution $\mu$ over forests of $G$ is not \textit{homogeneous} i.e. the support of $\mu$ contains different-size subsets of $E$, so we cannot immediately apply \cref{thm:main}. Let $\bar{\mu}$ be the \textit{complement} distribution of $\mu$ i.e. $\bar{\mu} (E \setminus F) =\mu(F)$ if $F$ is a forest, then sampling from $\mu$ and from $\bar{\mu}$ are equivalent. We add auxiliary elements to each $\overline{F} \in \supp(\bar{\mu})$ to obtain a homogeneous distribution. More precisely, using the \textit{Lorentzian polynomial} framework recently developed by \textcite{BH19}, we design a homogeneous distribution $\mu^{\uparrow}: \binom{E\cup Y}{n} \to \R_{\geq 0}$
whose projection to $E$ is $\bar{\mu}$ i.e. $\P_{T \sim \mu^{\uparrow}}{ T \cap E} = \bar{\mu} (T \cap E)$ where $Y$ is the set of auxiliary elements, such that the generating polynomial of $\mu^{\uparrow}$ is log-concave. Specifically, in \cref{lem:complementIndepPoly}, we prove that for any matroid $\matroid$ of rank $r$ over ground set $[n],$ the polynomial $f_{\matroid}(z_0, z_1, \cdots, z_n) = \sum_{S\in I(\matroid)} z_0^{\abs{S}} z^{[n] \setminus S}$ is Lorentzian, then use \textit{polarization} (see \cref{prop:polarization}) to transform $f_{\matroid}$ into a multi-affine homogeneous log-concave polynomial
\[f_{\matroid}^{\uparrow} (y_1, \cdots, y_r, z_1, \cdots, z_n) =  \sum_{S\in I(\matroid), T\in \binom{[r]}{\abs{S}}} \frac{1}{\binom{r}{\abs{S}}}   y^T z^{[n] \setminus S}\]
 That $f_{\matroid}$ is Lorentzian (or equivalently, completely log-concave) was not previously known, and could be of independent interest.
 
The distribution $\mu^{\uparrow}$ is generated by $f_{\matroid}^{\uparrow}.$ Our algorithm runs the down-up random walk on $\mu^{\uparrow}$, which mixes fast by \cref{thm:main}, then outputs $E \setminus  (T_t \cap E)$ where $T_t\in \supp(\mu^{\uparrow})$ is the random set we obtained after $t = O(n \log \frac{n}{\epsilon})$ down-up steps. Each step of the walk, even in the weighted case, can again be implemented in amortized $O(\log n)$-time using link-cut trees \cite{ST83, RTF18}. 

If we only consider the effect of the down-up walk on $\overline{T}_{t,E}: = E \setminus(T_t \cap E)$, then each step of the down-up walk can be viewed as follows: 
\begin{itemize}
    \item With probability $1-\frac{ \abs{\overline{T}_{t,E}} }{n}$, sample an edge $e \not\in \overline{T}_{t , E}$ uniformly at random and add $e$ to $\overline{T}_{t, E}$ 
    \item  
    If there is a cycle formed in $\overline{T}_{t, E}$ by the previous operation, remove an edge uniformly at random from the cycle. Else, with probability $\frac{q}{1+q}$, remove an edge uniformly at random from $\overline{T}_{t,E}.$ Note that this has no effects if $\overline{T}_{t,E}$ is already empty. 
\end{itemize}
Observe that if $q=0$, we never remove an edge from $\overline{T}_{t,E}$ unless $\overline{T}_{t,E}$ contains a cycle, thus if $\overline{T}_{0,E}$ is a spanning tree then so is $\overline{T}_{t,E}$ for all $t$. For $q=0$, our algorithm (to sample random spanning tree) is same as the one proposed by \textcite{RTF18}. Despite not having the tight mixing time analysis, they empirically observed fast mixing times for the proposed algorithm, and additionally showed how link-cut trees can be used to implement each step.

To prove \cref{thm:constrainedForests}, we only need to show that the following polynomial is Lorentzian
\[f_{\matroid}^{k_1} (z_0, z_1, \cdots, z_n)= \sum_{S \in \mathcal{I}(\matroid), \abs{S} \geq k_1} z_0^{\abs{S}} z^{[n]\setminus S}  \]
then set $\matroid$ to be the matroid whose bases are the size-$k_2$ forests of graph $G.$ Next, we employ the \textit{polarization} trick then running down-up walk framework which we use to prove \cref{thm:forests}. 
\subsection{Structure of the Paper}
In \cref{sec:prelim}  we provide some background on Markov chains and geometry of polynomials. In \cref{sec:mixing} we prove \cref{thm:main}. In \cref{sec:exchange} we prove certain approximate exchange properties, and their algorithmic implications. In \cref{sec:strees} we prove \cref{thm:forests,thm:constrainedForests}, as well as other results on sampling independent sets of an arbitrary matroid. The results in \cref{sec:strees} are mostly disjointed from \cref{sec:mixing,sec:exchange}.
\subsection{Acknowledgements}
The first author thanks Daniel Sleator and Gary L.\ Miller for insightful questions that led to the result on sampling spanning trees.

	\section{Preliminaries}
\label{sec:prelim}

We use $[n]$ to denote the set $\set{1,\dots,n}$ and $\binom{[n]}{k}$ to denote the family of size $k$ subsets of $[n]$. When $n$ is clear from context, we use $\1_S\in \R^n$ to denote the indicator vector of the set $S\subseteq [n]$, having a coordinate of $0$ everywhere except for elements of $S$, where the coordinate is $1$. We use $\conv$ to denote the operator that maps a set of points to their convex hull.

We use $z_S$ as shorthand for $\set*{z_i \given i \in S}$ and $z^S$ as shorthand for $\prod_{i\in S} z^i.$
For polynomial $f = \sum_{S} c_S z^S \in \C[z_0, \cdots, z_n],$ we let the support of $f$ be $\supp(f) := \set*{S: c_S \neq 0}$, and write $\partial_i f$ as shorthand for $\frac{\partial f}{\partial z_i}.$

We use $e_k(z_1, \cdots, z_n)$ to denote the $k$-th symmetric polynomial in $z_1, \cdots, z_n.$ We sometimes abuse notation and write $e_k ( u, z_{S})$ to denote the $k$-th symmetric polynomial in variables $z_S \cup \set*{u}.$
\subsection{Matroids}

In this paper we use one of the many cryptomorphic definitions of a matroid in terms of the polytope of its bases. For equivalence to other prominent definitions of a matroid, and more generally references to facts stated here see \cite{Oxl06}.
\begin{definition}\label{def:matroid}
	We say that a family $\B\subseteq \binom{[n]}{k}$ is the family of bases of a matroid if the polytope $\conv\set{\1_B\given B\in\B}$
	has only edges of the minimum possible length, namely $\sqrt{2}$. We call $k$ the rank of the matroid, and $[n]$ the ground set of the matroid.
	
	We let the family of independent sets of the matroid be $\mathcal{I} = \set*{I \in 2^{[n]} \given \exists B \in \B: I \subseteq B }$
\end{definition}

A well-known fact about matroids, that can be easily derived from \cref{def:matroid}, is that the dual of a matroid, defined below, is another matroid.
\begin{proposition}\label{prop:dual}
	If $\B\subseteq \binom{[n]}{k}$ is the family of bases of a matroid, then the following is also the family of bases of another matroid, called the dual matroid:
	\[ \B^*:=\set{[n]-B\given B\in \B}. \]
\end{proposition}

In this paper we will use a famous class of matroids constructed from graphs, called graphic matroids.

\begin{proposition}\label{prop:graphic}
	Let $G=(V, E)$ be a graph. Then the following is the family of bases of a matroid, called the graphic matroid of $G$:
	\[ \set{T\subseteq E\given T\text{ forms a spanning forest}}. \]
\end{proposition}
Note that the rank of the graphic matroid is $\leq \card{V}-1$ and the ground set is $E$. If $G$ is connected, then the bases are spanning trees of $G$, and the rank is exactly $\card{V} -1.$

\subsection{M-Convex Sets}
\begin{definition}[M-convex sets] \label{def:m-convex}
We define a subset $J \subseteq \N
^n$ to be M-convex if it satisfies any one of the following equivalent conditions:
\begin{itemize}
    \item  For any $\alpha
    ,\beta\in J$ and any index $i$ satisfying $\alpha_i > \beta_i$, there is an index $j$ satisfying
$\alpha_j < \beta_j$ and $\alpha- e_i + e_j \in  J.$
\item For any $\alpha, \beta \in J$ and any index $i$ satisfying $\alpha_i > \beta_i$
, there is an index $j$ satisfying
$\alpha_j < \beta_j$ and $\alpha -e_i + e_j \in J$ and $\beta - e_j + e_i \in J.$
\end{itemize}
We note that any M-convex set $J$ must be a subset of $\Delta_n^d:=\set*{\alpha \in \N^n \given \abs{\alpha}_1 = d}$ for some fixed $d$. Conversely, for $d=1$, any $J \subseteq \Delta_n^1$ is M-convex. 
\end{definition}
\subsection{Stable Polynomials}
\begin{definition}[Half-plane stable]
	Consider an open half-plane $H_{\theta} = \set*{e^{-i\theta} z \given \Im(z) > 0} \subseteq \C.$ We say a polynomial $g(z_1,\cdots, z_n) \in \C[z_1, \cdots, z_n]$ $H_{\theta}$-stable if $g$ does not have root in $H_{\theta}^n.$ In particular, the zero polynomial is $H_{\theta}$-stable.
	
	We call $H_0$ and $H_{\pi/2}$ the upper-half and right-half plane respectively. 
	We say $g$ is Hurwitz stable if it is $H_{\pi/2}$-stable. We say $g$ is real stable if it is $H_0$-stable and has real coefficients.
	
	We observe that for homogeneous polynomials, the definition of $H_{\theta}$-stable is equivalent for all angles $\theta.$
\end{definition}
A distribution $\mu: 2^{[n]} \to \R_{\geq 0}$ is strongly Rayleigh if and only if its generating polynomial is real stable \cite[see][]{BBL09}.
	
	Real stability is preserved under differentiation, and identification.
	\begin{theorem} [{\cite[Lemma 2.4]{wagner2009multivariate}}] \label{thm:real stable derivative}
	    If $g \in \R[z_1, \cdots, z_n]$ is real stable, then the following are also real stable
	    \begin{itemize}
	        \item $g\mid_{z_i = a}$ for $a\in \R$
	        \item $\partial_i g$ for all $i\in [n].$ 
	    \end{itemize}
	    
	\end{theorem}
	We will need the following classical fact \cite[see e.g.][Proposition 3.1 for a proof]{BBL09}.
\begin{theorem} \label{thm:symmetric} For $k \leq n$, the $k$-th symmetric polynomial in $n$ variables $e_k(z_1, \cdots, z_n)$ is real stable.
\end{theorem}

\subsection{Log-Concave Polynomials}

For a distribution or density function $\mu:\binom{[n]}{k}\to \R_{\geq 0}$ we denote by $g_\mu$ the generating polynomial of $\mu$ defined as
\[ g_\mu(z_1,\dots,z_n):=\sum_{S\in \binom{[n]}{k}} \mu(S)\prod_{i\in S} z_i. \]

We call a polynomial $g\in \R[z_1,\dots,z_n]$ with nonnegative coefficients log-concave when viewed as a function, it is log-concave over the positive orthant, i.e., for $x,y\in \R_{\geq 0}^n$ and $\lambda \in (0, 1)$
\[ g(\lambda x+(1-\lambda)y)\geq g(x)^\lambda g(y)^{1-\lambda}. \]
For a multiaffine polynomial $g$, its derivatives can be obtained as
\[ \partial_1 g=\lim_{c\to \infty}\frac{g(c,z_2,\dots,z_n)}{c}. \]
This shows that the derivatives of a multiaffine log-concave polynomial are limits of log-concave polynomials, which themselves are log-concave. It follows that a multiaffine homogeneous log-concave polynomial satisfies the seemingly stronger notions of \emph{strong log-concavity} \cite{Gur09} and \emph{complete log-concavity} \cite{AOV18, ALOV18, BH19}. The latter means that such polynomials are also closed under directional derivatives
\begin{lemma}\label{lem:clc}
	Let $g\in \R[z_1,\dots,z_n]$ be a multiaffine homogeneous polynomial with nonnegative coefficients. If $g$ is log-concave, then it is completely log-concave as well, which means that for any $k\in \Z_{\geq 0}$ and directions $v_1,\dots, v_k\in \R_{\geq 0}^n$, the following polynomial is log-concave:
	\[ \partial_{v_1}\cdots \partial_{v_k} g. \]
\end{lemma}


We will need this alternative characterization of completely log-concave homogeneous polynomial, which appears in \cite[Definition 2.6]{BH19} as \textit{Lorentzian polynomials}.

\begin{definition}[Lorentzian polynomials]
Let $g\in \R[z_1, \cdots, z_n]$ be a homogeneous polynomial of degree $d$ with nonnegative coefficients. We say $g$ is a Lorentzian polynomial if either $d\leq 1$, or $d\geq 2$, $\supp(g)$ is M-convex, and $\partial^{\alpha} g$ is real-stable for all $\alpha$ satisfying $\abs{\alpha}_1 = d-2.$
\end{definition}
In particular, if $g$ is Lorentzian then its support $\supp(g)$ is $M$-convex, and all its (directional) derivatives are Lorentzian.
\begin{proposition}[\cite{BH19}]  
    If $g$ is a Lorentzian polynomial, then for any $v \in \R_{\geq 0}$, $\partial_v g$ is Lorentzian.
\end{proposition}
\begin{theorem}[\cite{BH19}]\label{thm:lorentzian} 
    Let $g \in \R[z_1, \cdots, z_n]$ be a homogeneous polynomial with nonnegative coefficients. The following are equivalent:
    \begin{itemize}
    \item $g$ is completely log-concave
    \item $g$ is Lorentzian
    \end{itemize}
\end{theorem}

An important class of polynomials are those associated with uniform distributions over bases of a matroid.
\begin{theorem}[\cite{AOV18} based on \cite{AHK18}] \label{thm:basePoly}
	If $\B\subseteq \binom{[n]}{k}$ is the family of bases of a matroid, then the following polynomial is log-concave:
	\[ g(z_1,\dots,z_n):=\sum_{B\in \B} \prod_{i\in B} z_i. \]
\end{theorem}
\begin{theorem}[\cite{BH19,ALOV18}]\label{thm:indepPoly} 
    For any matroid $\matroid$ with family of independent sets $\mathcal{I}$, the polynomial
\[g_{\matroid}(y, z_1, \cdots , z_n) = \sum_
{I\in\mathcal{I}}
y^{n-\abs{I}} \prod_{i\in I} z_i\]
in $\R[y, z_1, \cdots, z_n]$ is completely log-concave.
\end{theorem}

We will use the following simple fact about log-concave polynomials:
\begin{proposition}[{\cite[see][]{ALOV18,BH18}}]\label{prop:oneeig}
	If $g$ is a log-concave polynomial with nonnegative coefficients, then $\nabla^2 g$ evaluated at any point in the positive orthant has at most one positive eigenvalue.
\end{proposition}

One of the basic operations preserving (complete) log-concavity is composition with a linear map. That is if $T:\R^m\to \R^n$ is an affine linear map for which $T(R_{\geq 0}^m)\subseteq \R_{\geq 0}^n$, then $g\circ T$ is (completely) log-concave as well. We state other operations that preserve Lorentzian property.
        \begin{proposition}[{\cite[Polarization]{BH19}}] \label{prop:polarization} 
    For an element $\kappa $ of $\N^n$ let
    \[\R_{\kappa}[z_1, \cdots, z_n] = \set*{\text{polynomials in $\R[z_i]_{1\leq i\leq n}$ of degree at most $\kappa_i$
in $z_i\forall i$}} \]
\[\R_{\kappa}^a[z_{ij}] = \set*{\text{
multi-affine polynomials in $\R[z_{ij}]_{1\leq i\leq n, 1 \leq j \leq \kappa_i}$}}\]
    The polarization map $\prod^{\uparrow}_{\kappa}$ is a linear map that sends monomial $z^{\alpha} = \prod_{i=1}^n z_i^{\alpha_i}$ to the product
    \[\frac{1}{\binom{\kappa}{\alpha}} \prod_{i=1}^n (\text{
elementary symmetric polynomial of degree $\alpha_i$
in the variables $\set{z_{ij}}_{1\leq j \leq \kappa_i}$}
)\]
where $\binom{\kappa}{\alpha} = \prod_{i=1}^n \binom{\kappa_i}{\alpha_i}.$
If $g \in \R_{\kappa}[z_i]_{1\leq i\leq n}$ is Lorentzian then $\prod^{\uparrow}_{\kappa}(g)$ is also Lorentzian. 
\end{proposition}
\begin{proposition} \label{prop:product}
    The product of two Lorentzian polynomials is also Lorentzian. 
\end{proposition}

\subsection{Down-Up Random Walk}

For two distributions $\nu, \mu$ we define the Kullback-Leibler divergence, KL-divergence for short, $\D{\nu\mid \mu}$ as
\[ \D{\nu \mid \mu}:=\ExX*{S\sim \mu}{\frac{\nu(S)}{\mu(S)}\log \frac{\nu(S)}{\mu(S)}}=\ExX*{S \sim \nu}{\log\frac{\nu(S)}{\mu(S)}},  \]
and the total variation distance between $\nu$ and $\mu$ as
\[ \norm{\nu-\mu}_\tv := \frac{1}{2}\sum_{S} \abs{\nu(S)-\mu(S)}. \]
The two are related by Pinsker's inequality:
\begin{proposition}[{\cite[see, e.g.,][]{CT12}}]\label{prop:pinsker}
	KL-divergence and the total variation distance are related by the following inequality
	\[ \norm{\nu-\mu}_\tv \leq \sqrt{\frac{1}{2}\D{\nu\mid \mu}} \]
\end{proposition}

\Textcite{CGM19} proved shrinkage of the KL-divergence under the down-up random walk. Coupled with Pinsker's inequality, this resulted in a mixing time bound.
\begin{lemma}{\cite{CGM19}}\label{lem:mlsi}
	If $\nu, \mu:\binom{[n]}{k}\to \R_{\geq 0}$ are distributions where $\mu$ has a log-concave generating polynomial, and $P$ is the down-up random walk operator whose stationary distribution is $\mu$, then
	\[ \D{\nu P\mid \mu P}=\D{\nu P\mid \mu}\leq (1-1/k)\D{\nu\mid \mu}. \]
\end{lemma}

	\section{Mixing Time Analysis}
\label{sec:mixing}

In this section we prove \cref{thm:main} by analyzing the down-up random walk for distributions $\mu$ that have a log-concave generating polynomial. As a reminder, in each step, the down-up random walk transitions from a set $S\in \binom{[n]}{k}$ to $S'\in \binom{[n]}{k}$ as follows:
\begin{itemize}
	\item From $S$ choose a subset $T\subseteq S$ of size $k-1$ uniformly at random.
	\item From all supersets $S'\supseteq T$, choose one with probability $\propto \mu(S')$.
\end{itemize}
Notice that the first step above simply drops a uniformly random element, and the second step replaces it with a new one (potentially the same element). Our high-level strategy is to prove that in $O(k\log k)$ steps, every element of the initial set is replaced at least once, and when this happens the distribution becomes a \emph{warm start} and converges to $\mu$ in an additional $O(k\log k)$ steps.

Let $\tau$ be the first time such that every element in our initial set has been replaced at least once. In other words think of initial elements as unmarked, and every time we replace an element we mark the new element brought in. Then $\tau$ is the first time that every element is marked.

We will prove the following:
\begin{lemma}\label{lem:mix}
	Let $S_t$ be the set at time $t$ in the down-up random walk. Then for any $X\in \binom{[n]}{k}$ and any time $t$,
	\[ \Pr{S_t=X\given \tau \leq t}\leq 2^{O(k^2)} \PrX{\mu}{X}. \]
\end{lemma}
Note that without $2^{O(k^2)}$, the r.h.s.\ is simply the stationary distribution. So this statement can be understood to say that as long as we have replaced each element at least once, we cannot be too far off from the stationary distribution.

Before proving \cref{lem:mix}, let us see finish the proof of \cref{thm:main} assuming it.

\begin{proof}[Proof of \cref{thm:main} assuming \cref{lem:mix}] 
	Note that for any fixed time $t$, we can simply bound $\Pr{\tau> t}$ by $k(1-1/k)^t\leq ke^{-t/k}$. In particular this probability rapidly converges to $0$ after about $k\log k$ steps.
	
	Now let $t_1<t_2$ be two time indices. Let $\nu_t$ denote the distribution of the state of random walk, i.e., $S_t$, at time $t$. Our goal is to bound $\norm{\nu_t-\mu}_\tv$, where for simplicity of notation, we assume $\mu$ is properly normalized to be a probability distribution. Let $\nu_t'$ be the distribution of $S_t$ conditioned on $\tau\leq t$, and let $\nu_t''$ be the distribution of $S_t$ conditioned on $\tau\geq t$. Then we can write
	\[ \nu_{t_1}=\Pr{\tau\leq t_1}\cdot \nu_{t_1}'+\Pr{\tau>t_1}\cdot \nu_{t_1}''. \]
	If $P$ denotes the random walk operator, then note that $\nu_{t_2}=\nu_{t_1}P^{t_2-t_1}$. So we get
	\[ \nu_{t_2}=\Pr{\tau\leq t_1}\nu_{t_1}'P^{t_2-t_1}+\Pr{\tau>t_1}\nu_{t_1}''P^{t_2-t_1}. \]
	Using the triangle inequality we can bound
	\[ \norm{\nu_{t_2}-\mu}_\tv \leq \norm{\nu_{t_1}'P^{t_2-t_1} - \mu}_\tv+\Pr{\tau>t_1}. \]
	Here we used the fact that $\Pr{\tau\leq t_1}\leq 1$, and $\norm{\nu_{t_1}''P^{t_2-t_1}-\mu}_\tv\leq 1$; the latter inequality is because $\norm{\cdot }_\tv$ is always upper bounded by $1$.
	
	We can bound the second term in the above inequality by $ke^{-t_1/k}$ as stated before. For the first term, note that the KL-divergence between $\nu_{t_1}'$ and $\mu$ is at most $O(k^2)$ by \cref{lem:mix}. This is because
	\[ \D{\nu_{t_1}'\mid \mu}=\ExX*{S\sim \nu_{t_1}'}{\log \frac{\nu_{t_1}'(S)}{\mu(S)}}\leq \log(2^{O(k^2)})=O(k^2). \]

	 So by \cref{lem:mlsi} in $t_2-t_1$ steps this KL-divergence decreases to $(1-1/k)^{t_2-t_1}O(k^2)=O(k^2e^{-(t_2-t_1)/k})$. By Pinsker's inequality, \cref{prop:pinsker}, we get that
	\[ \norm{\nu_{t_1}'P^{t_2-t_1}- \mu}_\tv\leq O(ke^{-(t_2-t_1)/2k}). \]
	So in the end we get the following bound
	\[ \norm{\nu_{t_2}-\mu}_\tv\leq O(ke^{-(t_2-t_1)/2k}+ke^{-t_1/k}). \]
	In order for this to be at most $\epsilon$, it is enough to make sure that $\min\set{t_1, t_2-t_1}=\Omega(k\log k+k\log \frac{1}{\epsilon})$. So we can simply let $t_1=t_2/2$, and then make sure that $t_2=\Omega(k\log (k/\epsilon))$.
\end{proof}

As the main tool we use to prove \cref{lem:mix}, we introduce a new inequality for log-concave polynomials, that we call approximate exchange. We state the inequality below and defer its proof to \cref{sec:exchange}.
\begin{lemma}\label{lem:exchange}
	Any $\mu:\binom{[n]}{k}\to \R_{\geq 0}$ with a log-concave generating polynomial satisfies a $2^{O(k)}$-exchange property. That is, for every $S, T\in \binom{[n]}{k}$ and $i\in S$ there exists $j\in T$ such that
	\[ \mu(S)\mu(T)\leq 2^{O(k)} \mu(S-i+j)\mu(T+i-j). \]
\end{lemma}

Armed with \cref{lem:exchange}, let us prove \cref{lem:mix}.
\begin{proof}[Proof of \cref{lem:mix}]	
	Let's look at the down-up walk process \emph{with orders}. This means that we start with some elements $e_1,\dots,e_k$ that together form the starting set. In each time step we replace one of the $e_i$'s. But we keep track of the ordering and do not convert these to sets. So we can talk about $e_i^t$ as the $i$-th element at time $t$. In particular $S_t$ is simply the \emph{unordered} collection $\set{e_1^t,\dots,e_k^t}$. Let's say that $X=\set{f_1,\dots,f_k}$. Then to have $S_t=X$, there must be some permutation of $f_1,\dots,f_k$ that equals $e_1^t,\dots,e_k^t$. We will show that for any such permutation the promised bound in \cref{lem:mix} holds. Since there are $k!=2^{O(k\log k)}$ many permutations, this extra factor of $k!$ can be absorbed into the factor of $2^{O(k^2)}$ without any loss. So we fix an arbitrary permutation, w.l.o.g.\ the identity permutation, and try to bound the following
	\[ \Pr{e_1^t=f_1,\dots,e_k^t=f_k\given \tau\leq t}. \]
	Since we are conditioning on $\tau\leq t$, note that there must be some time $\tau_i\leq t$, which is the last time before $t$ where the $i$-th element gets replaced by the down-up random walk. We will bound the above probability, even conditioned on $\tau_1,\dots,\tau_k$ having any set of fixed values up to $t$. Note that the index of the element that gets replaced in every step is uniformly random and independent of everything else that happens in the random walk, in particular the identity of the elements that come in as replacements. In the rest of the proof, we condition on the indices of the elements that get replaced at every step up to time $t$; note that this also uniquely determines $\tau_1,\dots, \tau_k$, so we assume $\tau_1,\dots,\tau_k$ are some fixed time indices. W.l.o.g.\ assume that $\tau_1<\tau_2<\dots<\tau_k$. We will use induction to prove the following statement for $i=0,\dots, k$:
	\[ \Pr{e_1^{\tau_1}=f_1,e_2^{\tau_2}=f_2,\dots,e_i^{\tau_i}=f_i\given \text{replacement indices}}\leq 2^{O(ik)}\PrX{U\sim \mu}{f_1,\dots,f_i\in U}. \]
	Notice that for $i=0$, both sides are trivially equal to $1$, and for $i=k$, this inequality is the main statement we want to prove. 
	
	It remains to show the inductive step. We will show that going from $i-1$ to $i$, the l.h.s.\ gets multiplied by a smaller quantity compared to the r.h.s. If we have below inequality in hand, then it is not hard to see that we can complete the induction, since the factors that get multiplied on each side are the two sides of this inequality.
	\[ \Pr{e_i^{\tau_i}=f_i\given e_1^{\tau_1}=f_1,\dots,e_{i-1}^{\tau_{i-1}}=f_{i-1}\text{ and replacement indices}}\leq 2^{O(k)}\PrX{U\sim \mu}{f_i\in U\given f_1,\dots,f_{i-1}\in U}. \]
	Instead of conditioning only on $f_1,\dots,f_{i-1}$ being chosen at the appropriate times on the l.h.s., we will refine the conditioning and condition on the \emph{history of the random walk} up to time $\tau_i-1$. This means we can in particular assume that the elements $e_{i+1}^{\tau_i},\dots,e_k^{\tau_i}$ are fixed, that $e_1^{\tau_i}=f_1,\dots,e_{i-1}^{\tau_i}=f_{i-1}$, and the only uncertain thing is what the $i$-th element is being replaced by at time $\tau_i$.
	
	Let $S=\set{f_1,\dots,f_i, e_{i+1}^{\tau_i},\dots,e_k^{\tau_i}}$. Then the conditional probability of choosing $f_i$ at time $\tau_i$ is:
	\[ \Pr{e_i^{\tau_i}=f_i\given e_1^{\tau_1}=f_1,\dots,e_{i-1}^{\tau_{i-1}}=f_{i-1}\text{ and replacement indices}} = \frac{\mu(S)}{\sum_{V\supset S-f_i}\mu(V)}. \]
	On the other hand 
	\[ \PrX{U\sim \mu}{f_i\in U\given f_1,\dots,f_{i-1}\in U}=\frac{\sum_{U\ni f_1,\dots,f_i}\mu(U)}{\sum_{T\ni f_1,\dots,f_{i-1}} \mu(T)} \]
	So we have to show the following:
	\[ \mu(S)\parens*{\sum_{T\ni f_1,\dots,f_{i-1}} \mu(T)}\leq 2^{O(k)}\parens*{\sum_{V\supset S-f_i}\mu(V)}\parens*{\sum_{U\ni f_1,\dots,f_i} \mu(U)}. \]
	We will give an injection from the terms on the l.h.s.\ to the terms in the expanded form of the r.h.s. Choose some set $T\ni f_1,\dots,f_{i-1}$. Apply \cref{lem:exchange} to $S$ and $T$ with the element $f_i\in S$. We get that there must be some element $e\in T$ such that
	\[ \mu(S)\mu(T)\leq 2^{O(k)}\mu(S-f_i+e)\mu(T+f_i-e). \]
	Note that $V:=S-f_i+e$ contains $S-f_i$, and $U:=T+f_i-e$ contains $\set{f_1,\dots,f_i}$. So $\mu(U)\mu(V)$ appears on the r.h.s.\ of the desired inequality. So for each $T$ appearing on the l.h.s.\ of the desired inequality we produced a pair of $U$ and $V$. Note that this mapping from $T$ to $(X,Y)$ is injective. This is because given $(X,Y)$, we can recover $T$ as the xor/symmetric difference of the other three sets, that is $T=S\Delta U\Delta V$.
\end{proof}

	\section{Approximate Exchange Property}
\label{sec:exchange}

In this section we prove \cref{lem:exchange}. In addition, we show other variant(s) of approximate exchange property (\cref{lemma:rsExchange}) with implication for the approximation guarantee of local search and greedy on the problem of $\max_{S} \mu(S)$ given a log-concave/strongly-Rayleigh distribution $\mu:\binom{[n]}{k}\to\R_{\geq 0}$ (\cref{lem:greedy,lem:localsearch}).

\begin{definition}
	We say that $\mu:\binom{[n]}{k}\to\R_{\geq 0}$ has an $\alpha$-approximate exchange property, or $\alpha$-exchange for short, if for every $S, T\in \binom{[n]}{k}$ and every $i\in S$, there exists $j\in T$ such that
	\[ \alpha\cdot \mu(S-i+j)\mu(T+i-j)\geq \mu(S)\mu(T). \]
\end{definition}
Note that if $\mu$ is the indicator of bases of a matroid, then it has a $1$-exchange property, also known as the strong basis exchange property \cite{Oxl06}.

Although we do not directly need it, we give another example where approximate exchange can be proven by elementary means. This is the class of $k$-determinantal point processes \cite{BBL09, KT12}.
\begin{proposition} \label{prop:kDPP}
	Suppose that $\mu:\binom{[n]}{k}\to\R_{\geq 0}$ is defined as
	\[ \mu(S)=\det([v_i]_{i\in S})^2, \]
	for some vectors $v_1,\dots,v_n\in \R^k$. Then $\mu$ has a $k^2$-exchange property.
\end{proposition}
\begin{proof}
	It is enough to consider the case where $S$ and $T$ are disjoint; otherwise, the problem can be reduced to lower values of $k$ by taking out the intersection, and projecting all vectors on the orthogonal complement of the space spanned by the intersection.
	
	Define the number $\beta_j$ as $\sqrt{\mu(S-i+j)\mu(T+i-j)}$ and let $\alpha$ be $\sqrt{\mu(S)\mu(T)}$. The Pl\"ucker relations for the Grassmanian \cite[see, e.g.,][]{Abe80} say that a signed sum of $\alpha$ and $\beta_j$ is zero:
	\[ \alpha+\sum_{j\in T}\pm \beta_j=0. \]
	This means that there is at least one $j$ such that $\abs{\beta_j}\geq \frac{1}{k}\alpha$, and this concludes the proof.
\end{proof}

Next we take steps to prove \cref{lem:exchange}, namely that if $\mu:\binom{[n]}{k}\to\R_{\geq 0}$ has a log-concave generating polynomial $g_\mu$, then $\mu$ has a $2^{O(k)}$-exchange property. We conjecture that a $k^{O(1)}$-exchange property should hold, but even if true, this will not improve the mixing time results in this paper beyond constants hidden in the $O(\cdot)$ notation.

Our strategy is to prove the case of $k=2$ of \cref{lem:exchange} by using log-concavity of $g_\mu$ (note that $k=1$ is trivial). We will then use an induction to prove the general case. We remark that this type of induction is a standard procedure used in many other places, such as in the context of proving Pl\"ucker relations and $M^\natural$-concavity \cite{MS18}.

Before delving into the proof, note that we can always assume $S\cap T=\emptyset$. This is because we can always condition the distribution $\mu$ on having any set of elements, and then throwing out those elements; this operation corresponds to taking partial derivatives of $g_\mu$ which results in a log-concave polynomial by \cref{lem:clc}. In particular, we can condition $\mu$ on having $S\cap T$, and then throwing out $S\cap T$ from the ground set.

\begin{proof}[Proof of \cref{lem:exchange} for $k=2$]	
	When $k=2$, we might as well assume that $n=4$, because no element outside of $S\cup T$ is important, and we can condition the distribution $\mu$ on not having those elements. This corresponds to substituting $0$ for variables outside $S\cup T$ in $g_\mu$ which preserves log-concavity.
	
	So our goal now is to show that for a log-concave quadratic polynomial in four variables
	\[ g_\mu = \sum_{\set{i, j}\in \binom{[4]}{2}} \mu(\set{i, j})z_iz_j, \]
	we have an $O(1)$-exchange property. W.l.o.g.\ assume that $S=\set{1,2}$ and $T=\set{3,4}$.
	
	Let us consider $\nabla^2 g_\mu$. This is a constant matrix, which has at most one positive eigenvalue by \cref{prop:oneeig}. On the other hand it is a matrix with nonnegative entries, so it must have at least one nonnegative eigenvalue as well. Analyzing the possible signs of the eigenvalues, we see that their product, i.e., the determinant is nonpositive:
	\[ \det(\nabla^2 g_\mu)\leq 0. \]
	This determinant can be written in a special way. Let us define:
	\[ A:=\mu(\set{1,2})\mu(\set{3,4}), \]
	\[ B:=\mu(\set{1,3})\mu(\set{2,4}), \]
	\[ C:=\mu(\set{1,4})\mu(\set{2,3}). \]
	Notice that approximate exchange for $S, T$ any any $i\in S$ is equivalent to saying that $A\leq O(1)\cdot \max\set{B, C}$.
	We can write $\det(\nabla^2 g_\mu)=A^2+B^2+C^2-2(AB+AC+BC)$. So we get the inequality 
	\[A^2+B^2+C^2\leq 2(AB+AC+BC).\]
	This is the same as
	\[ (A-B-C)^2\leq 4BC. \]
	Taking square-roots we get
	\[ A-B-C\leq 2\sqrt{BC}, \]
	which is the same as saying
	\[ A\leq (\sqrt{B}+\sqrt{C})^2. \]
	Taking square-roots again we get
	\[ \sqrt{A}\leq \sqrt{B}+\sqrt{C}. \]
	In particular one of $\sqrt{B}$ and $\sqrt{C}$ must be at least $\frac{1}{2}\sqrt{A}$. This proves that $\mu$ satisfies a $2^2=4$-approximate exchange property for $S=\set{1,2}$ and $T=\set{3,4}$.
\end{proof}

We now complete the proof by inducting on $k$.
\begin{proof}[Proof of \cref{lem:exchange} for the general case]
	We can assume that for any $S, T$ such that $\card{S\cap T}\geq 1$, we have a $2^{O(k-\card{S\cap T})}$-approximate exchange property. This is because by the arguments we had, such nonempty intersections can be reduced to smaller values of $k$ by conditioning and throwing out $S\cap T$.
	
	Now let $S\cap T=\emptyset$ and let $i\in S$ be given. Our goal is to find $j$ such that
	\[ \mu(S)\mu(T)\leq 2^{O(k)} \mu(S-i+j)\mu(T+i-j). \]
	
	Let $i'\neq i$ be another, arbitrary, element of $S$. We will exchange $i'$ with an element $j'\in T$ and use induction on $S-i'+j'$ and $T$. We need to be careful how we choose $j'$ though. Let us choose $j'$ to be the element of $T$ that maximizes the expression $\mu(T+i-j')\mu(S-i'+j')$. The reason for this choice will become apparent in the rest of he proof.
	
	Then the sets $S-i'+j'$ and $T$ have an intersection of one element, so by induction we know an approximate exchange property for them. Therefore, there must be a $j\in T$ such that
	\begin{equation}\label{eq:eq1}
		\mu(S-i'+j')\mu(T)\leq 2^{O(k-1)}\mu(S-i-i'+j+j')\mu(T+i-j).
	\end{equation} 
	We will apply approximate exchange a second time. The sets $S$ and $S-i-i'+j+j'$ have a very large intersection. In particular their exchange property reduces to the case of $k=2$ of \cref{lem:exchange}, which we have already proven. By this exchange property, we have
	\begin{equation}\label{eq:eq2}
	 \mu(S)\mu(S-i-i'+j+j')\leq 2^{O(1)} \max\set*{\mu(S-i+j)\mu(S-i'+j'), \mu(S-i+j')\mu(S-i'+j) }.
	\end{equation}
	If the first term in \cref{eq:eq2} achieves the maximum, then we are done, because multiplying \cref{eq:eq1,eq:eq2} yields
	\begin{multline*} \mu(S-i'+j')\mu(T)\mu(S)\mu(S-i-i'+j+j')\leq \\ 2^{O(k)} \mu(S-i-i'+j+j')\mu(T+i-j)\mu(S-i+j)\mu(S-i'+j'), \end{multline*}
	which simplifies to 
	\[ \mu(S)\mu(T)\leq 2^{O(k)}\mu(S-i+j)\mu(T+i-j), \]
	showing that $i$ can be exchange for $j$.
	
	So assume that the second term in \cref{eq:eq2} achieves the maximum. We will show that in this case $i$ can be exchanged for $j'$. Multiplying \cref{eq:eq1,eq:eq2} yields
	\begin{multline*} \mu(S-i'+j')\mu(T)\mu(S)\mu(S-i-i'+j+j')\leq \\ 2^{O(k)} \mu(S-i-i'+j+j')\mu(T+i-j)\mu(S-i+j')\mu(S-i'+j), \end{multline*}
	which simplifies to
	\[ \mu(S)\mu(T)\leq 2^{O(k)}\cdot \mu(S-i+j')\mu(T+i-j')\frac{\mu(T+i-j)\mu(S-i'+j)}{\mu(T+i-j')\mu(S-i'+j')} \]
	Notice that by our choice of $j'$, the fraction appearing on the r.h.s.\ is $\leq 1$. So we can conclude that
	\[ \mu(S)\mu(T)\leq 2^{O(k)}\mu(S-i+j')\mu(T+i-j'). \]
	\end{proof}
	If we require the stricter assumption that $\mu$ is generated by a real-stable polynomial, then we get a $k^2$-exchange. This is a generalization of \cref{prop:kDPP}.
	\begin{lemma} \label{lemma:rsExchange}
		Consider $\mu:\binom{[n]}{k}\to \R_{\geq 0}$ that is generated by a real-stable polynomial.  For every $S, T\in \binom{[n]}{k}$ and $i\in S\setminus T$ 
		\begin{equation}\label{ineq:rsExchange1}
		    \sqrt{\mu(S) \mu(T)} \leq  \sum_{j\in T\setminus S} \sqrt{\mu(S-i+j)\mu(T+i-j)}
		\end{equation}
		Consequently, there exists $j\in T\setminus S$ such that
		\begin{equation}\label{ineq:rsExchange2}
		    \mu(S)\mu(T)\leq k^2 \mu(S-i+j)\mu(T+i-j)
		\end{equation}
		Thus $\mu$ satisfies a $k^2$-exchange property.
Moreover, for $ S\in \binom{[n]}{k}$ and $j\not \in S$,
    \begin{equation} \label{ineq:rsExchange3}
    \mu(S) \mu(j) \leq k\sum_{e \in S} \mu(S+j -e) \mu(e)
\end{equation}
where $\mu(t) = \sum_{T\in \binom{[n]}{k}: t \in T}\mu(T) $ for $t \in \set*{j, e}.$
	\end{lemma}
	We need the following theorem about single variate Hurwitz stable polynomial, due to \cite{asner70}.
	\begin{theorem}
	Consider Hurwitz-stable polynomial $f(z) = \sum_{i=0}^n a_i z^i$ with $a_i \geq 0 \forall i.$ Its Hurwitz matrix $H = (h_{ij}) \in \R^{n\times n}$ is defined by $h_{ij} = a_{2j-i}$ for $0 \leq 2j-i \leq n$, otherwise $h_{ij}=0.$ $H$ is totally nonnegative, in the sense that all its minors are nonnegative.
	\end{theorem}
	As an immediate consequence, we obtain the following lemma about coefficients of single variate Hurwitz stable polynomial.
	\begin{lemma}\label{lemma:hurwitzCoeff}
	Consider Hurwitz-stable polynomial $f(z) = \sum_{i=0}^{2t-1} a_i z^i$ with $a_i\geq 0 \forall i.$ Then $a_{2t-1} a_0 \leq a_{2t-2} a_1. $  
	\end{lemma}
    	\begin{proof}
    	    By total-nonnegativivity of the Hurwitz matrix $H$, we have 
    	    \[\det\begin{bmatrix}
    	    h_{1,1} & h_{1, t}\\
    	    h_{2,1} & h_{2, t}
    	    \end{bmatrix} = \det \begin{bmatrix}
    	    a_1 & a_{2t-1}\\
    	    a_0 & a_{2t-2}
    	    \end{bmatrix} = a_1 a_{2t-2} - a_0 a_{2t-1}  \geq 0
\]    	
\end{proof}
We are ready to prove \cref{lemma:rsExchange}. The idea is to construct a Hurwitz stable polynomial whose coefficients correspond to the LHS and RHS of \eqref{ineq:rsExchange1}, then use \cref{lemma:hurwitzCoeff} to derive \eqref{ineq:rsExchange1}.
	\begin{proof}[Proof of \cref{lemma:rsExchange}]
	We first show that \eqref{ineq:rsExchange1} implies \eqref{ineq:rsExchange2} and \eqref{ineq:rsExchange3}. Indeed,
	\[\sqrt{\mu(S) \mu(T)} \leq  \sum_{j\in T\setminus S} \sqrt{\mu(S-i+j)\mu(T+i-j)} \leq k \max_{j \in T \setminus S} \sqrt{\mu(S-i+j)\mu(T+i-j)} \]
	For $j\not\in S$ and $T\in \binom{[n]}{k}$ containing $j$, using \eqref{ineq:rsExchange1} and Holder's inequality, we get
	\[ \mu(S) \mu(T) \leq \left(\sum_{e\in S\setminus T} \sqrt{\mu(S-e+j)\mu(T+e-j)}\right)^2\leq k \sum_{e\in S\setminus T} \mu(S-e+j)\mu(T+e-j) \]
	Summing over all such $T$, while observing that $\sum_{T\in \binom{[n]}{k}: j \in T} \mu(T+e-j) \leq \mu(e)$, gives
	\[\mu(S) \mu(j) =\sum_{T\in \binom{[n]}{k}: j \in T}\mu(S) \mu(T) \leq  k \sum_{T\in \binom{[n]}{k}: j \in T} \sum_{e\in S\setminus T} \mu(S-e+j)(\mu(T+e-j)\leq k \sum_{e\in S} \mu(S-e+j)\mu(e)\]
	    Any $\mu:\binom{[n]}{k}\to \R_{\geq 0}$ with a real-stable generating polynomial can be approximated by a strictly real stable $\tilde{\mu}: \binom{[n]}{k} \to \R_{>0}. $, in the sense that $\abs{\tilde{\mu} (S) - \mu(S)} < \epsilon$ where $
	   \epsilon$ can be made arbitrarily small. This statement appears in \cite{Nui68} and \cite[Prop. 2.2]{BH19}. 
	    We can prove the lemma for $\tilde{\mu}$ then take $\epsilon \to 0$ to get the corresponding inequality for $\mu$. Thus, we can assume $\mu(S) > 0 \forall S \in \binom{[n]}{k}.$
	   
	    We deal with the case when $S\cap T = \emptyset$ and $[n] = S \cup T.$ Other cases can be reduced to this scenario by setting $z_i$ to $0$ for $i\not \in S \cup T$, and taking derivative with respect to $i\in S\cap T.$ Let $t: =\abs{S} = \abs{T}$ then $n=2t$. We can rewrite $f$ as
$$f(z_1, \cdots, z_{2t}) = \sum_{W \in \binom{[2t]}{t}} \mu(W) z_W$$ with $\mu(W) > 0 \forall W.$ 

For $W \in \binom{[2t]}{t}$ let $\Delta(W) = \sqrt{\mu(W) \times \mu([2t] \setminus W)}.$

Fix $e\in T$. We want to show
\[\sum_{i\in S} \Delta(S +e -i) \geq \Delta(S). \]

Since $f$ is homogeneous and real stable, it is also Hurwitz stable.

In $f$, set 
\begin{enumerate}
    \item $z_{e} = 1,$
    \item $\forall i\in S: z_i = z^{-1}\delta_i$ with $\delta_i = \sqrt{\frac{\mu(S+e-i)}{\mu(T-e+i)} }, $ 
    \item Let $B =  \prod_{i\in S} \delta_i > 0.$
    $\forall j \in T\setminus e: z_j = z B^{\frac{1}{t-1}}.$
\end{enumerate}
 and multiply $f$ by $B^{-1} z^t$, we obtain Hurwitz stable $\tilde{f}(z)$ with positive coefficients and degree $2t-1.$ 
  We rewrite $\tilde{f}(z) = b_0 z^{2t-1} + b_1 z^{2t-2} + \cdots + b_{2t-2} z_1 + b_{2t-1}. $
  
 Note that the monomial $z^W \mu(W) $ in $f$ contributes to $b_{2t-2}$ iff $\abs{S \Delta W} = 2$ and $e\in W$ i.e. $W = S-i +e$ for some $i\in S.$ Similarly, $z^W \mu(W)$ contributes to $b_1$ in iff $\abs{S \Delta W} = 2t-2$ and $e \not\in W$ i.e. $W = T-e + i$ for some $i\in S.$
 
 A routine calculation gives
 \begin{align*}
      b_{2t-1}& = \mu(S) B^{-1} \prod_{i\in S} \delta_i = \mu(S)  \\
     b_0 &= \mu(T)  \prod_{j \in T-e} B^{\frac{1}{t-1}}  = \mu(T)   \\
     b_{2t-2} &= B^{-1} \sum_{i\in S} \left(\mu(S-i+e) \prod_{j\in S-i} \delta_j \right) 
     = \sum_{i\in S}\mu(S-i+e) \delta_i^{-1} = \sum_{i\in S} \Delta(S+e-i)
     \\
     b_1 &= B^{-1}\sum_{i\in S}\left( \mu(T-e+i) \delta_i \prod_{j\in T\setminus e} B^{1/(t-1)}\right)
     = \sum_{i\in S} \Delta(S+e-i) 
 \end{align*}
Since $\tilde{f}$ is Hurwitz stable with nonnegative coefficients, \cref{lemma:hurwitzCoeff} implies $b_1 b_{2t-2} \geq b_0 b_{2t-1}$ i.e.
$(\sum_{i\in S} \Delta(S+e-i))^2 \geq \Delta(S)^2$ as required.
	\end{proof}
A consequence of \cref{lemma:rsExchange} is that a natural local search and greedy algorithms for finding the maximum of $\mu: \binom{[n]}{k} \to \R_{\geq 0}$ give a $(k!)^2\simeq k^{O(k)}$-approximation of $\max\set{ \mu(S)}$, assuming $\mu$'s generating polynomial $g_{\mu}(z_1, \cdots, z_n ) = \sum_{S} \mu(S) z^S$ is real-stable/strongly Rayleigh. This generalizes similar results for determinantal $\mu$ to the class of strongly Rayleigh distributions \cite[see, e.g.,][]{KD16}, giving further evidence for the efficacy of local search and greedy methods \cite[see also][]{Fed13}.

For subset $T$ of $[n]$ of size $\leq k$, let $\mu(T) = \sum_{S \in  \binom{[n]}{k}: S \supseteq T} \mu(S).$

\begin{algorithm} 
\SetAlgoLined
Initialize $S \leftarrow \emptyset$

While $\abs{S} < k:$
Pick $i\not\in S$ that maximizes $\mu(S\cup \set*{i}),$ and update $S \leftarrow S \cup \set*{i}$ 
\caption{Greedy} \label{alg:greedy}
\end{algorithm}
\begin{lemma}\label{lem:greedy}
If $\mu$ is strongly Rayleigh, then the output $S \in \binom{[n]}{k}$ of \cref{alg:greedy} is a $(k!)^2$-approximation of  $\OPT: = \max_{T\in \binom{[n]}{k} } \mu(T).$
\end{lemma}
\begin{proof}
    W.l.o.g., assume $\mu$ is not identically $0.$
For $j\in [k],$ let $i_j$ be the element added to $S$ at the $j$-th iteration of the while loop. Let $S_0 = \emptyset$, $S_j = S_{j-1} \cup \set*{i_j}$ and $\OPT_j = \argmax_{T\in \binom{[n]}{k}: T \supseteq S_j } \mu(T).$ Note that $\mu(\OPT_0) = \OPT$ and $\OPT_k = S.$ We show by induction on $j \in [k]$ that $\mu(\OPT_{j-1}) \leq (k-j+1)^2\mu(\OPT_j) ,$ then conclude that $\OPT \leq (k!)^2 \mu(S).$

First, observe that $\mu(S_j) > 0$ for all $j\in \set*{0, \cdots, k}.$ For $j= 0$, this is trivially true since $\mu(\emptyset) = \sum_{S} \mu(S).$ For $j\geq 1,$ \[\mu(S_j ) = \max_{i\not \in S_{j-1}} \mu(S_{j-1} \cup \set*{i}) \geq \frac{1}{k} \sum_{i\not \in S_{j-1}}\mu (S_{j-1} \cup \set*{i})\geq \frac{1}{kn} \mu(S_{j-1}) > 0 \] 
Consider $\mu_{j-1}: \binom{[n]\setminus S_{j-1}}{k+1-j} \to \R_{\geq 0}$ defined by $\mu_{j-1}(T) = \mu(T \cup S_{j-1}).$ Observe that $\mu_{j-1}$ is generated by the real stable polynomial $\partial_{i_{j-1}} \cdots \partial i_1 f$ (see \cref{thm:real stable derivative}). Let $X = \OPT_{j-1} \setminus S_{j-1}.$ Apply \eqref{ineq:rsExchange1} in \cref{lemma:rsExchange} to $X$ and $i_j \not \in S_{j-1}$, we have 
\[\mu_{j-1}(X) \mu_{j-1}(i_j) \leq (k+1-j) \sum_{e \in X} \mu_{j-1}(X+i_j - e) \mu_{j-1}(e) \leq  (k+1-j)^2\mu( \OPT_j)\, \mu_{j-1}(i_j) \]
where the last inequality follows from $i_j = \argmax_{i \not\in S_{j-1}} \mu(S_{j-1} \cup\set*{i})= \argmax_{i \not\in S_{j-1}}\mu_{j-1}(i),$ $\OPT_j = \argmax_{T\in \binom{[n]}{k}: T \supseteq S_j } \mu(T)$ and the fact that $S_j \subseteq (X\cup S_{j-1}\cup \set*{i_j} \setminus \set*{e}) $ for $e\in X.$  Dividing both sides by $\mu_{j-1}(i_j) = \mu(S_{j-1} \cup \set*{i_j})> 0$ gives 
\[\mu(\OPT_{j-1}) = \mu_{j-1} (X) \leq (k+1-j)^2 \mu(\OPT_j)\]

\end{proof}
We remark that similar guarantees can be obtained for a closely related local search algorithm, which moves between sets of size $k$, each time replacing one element by another. Note that our improved exchange property for strongly Rayleigh distributions is crucial in obtaining $k^{O(k)}$-approximation. For arbitrary log-concave distributions, we can show the approximate exchange property in \cref{lemma:rsExchange} with approximation factor $2^{O(k)}$ instead of $k^2$, thus proving a $2^{O(k^2)}$-approximation guarantee for greedy. Furthermore, we show that local search yields $k^{O(k)}$-approximation. 
\begin{algorithm} 
\SetAlgoLined
Initialize $S \leftarrow S_0$ for some $S_0$ with $\mu(S_0) \neq 0$

While $\mu(S) < \alpha\cdot  \mu(S-i+j)$ for some $j\not\in S$ and $i\in S$ do:
Update $S \leftarrow \argmax_{S'\in \set*{S-i+j \given \, j\not\in S, i\in S }} \mu(S')$ 
\caption{$\alpha$-local search ($\alpha \leq 1$) }\label{alg:localsearch}
\end{algorithm}
\begin{lemma} \label{lem:localsearch}
If $\mu$ is log-concave, then the output $S \in \binom{[n]}{k}$ of \cref{alg:localsearch} is a $(k!)^2/\alpha^k$-approximation of  $\OPT: = \max_{T\in \binom{[n]}{k} } \mu(T).$
\end{lemma}
\begin{proof}
We sketch a proof for the case $\alpha=1.$ The proof for general $\alpha \leq 1$ is entirely analogous.
The main ingredient in our proof is an inequality of the form:
\[\mu(S) \mu(T) \leq \frac{1}{\abs{S \setminus T}^2}\left(\sum_{i \in S\setminus T, j \in T\setminus S} \mu(S-i+j) \right) \left(\sum_{i \in S\setminus T, j \in T\setminus S} \mu(T+i-j)\right)\]
For $S \cap T = \emptyset,$ the first inequality follows from the complete-log-concavity of the polynomial $\tilde{f}_{\mu}( y, x) = \mu(S) y^k + \left(\sum_{i\in S, j\in T} \mu(S-i+j)\right) y^{k-1} x + \cdots +\left(\sum_{i\in S, j\in T} \mu(T+i-j)\right) y x^{k-1} + \mu(T) x^k$
obtained by setting $z_i = y$ for $i\in S$, $z_j = x$ for $j \in T$ and $z_i = 0$ otherwise, in the generating polynomial $f_{\mu}$ \cite[see e.g.][]{BH19}. We can reduce the remaining cases to the case $S\cap T = \emptyset$ by taking derivative $\prod_{i\in S \cap T}\partial_i$ of $f_{\mu}.$

When $S$ is a local-maxima, we get
 \[\mu(S) \mu(T)  \leq k^2 \mu(S) \max_{i \in S\setminus T, j \in T\setminus S} \mu(T+i-j)\]
Successively apply this inequality, first with $T = T_0 :=\arg\max_{T'} \mu(T')$ then with \[T_{\ell}: = \argmax_{i \in S\setminus T_{\ell-1}, j \in T_{\ell-1}\setminus S} \mu(T_{\ell-1}+i-j) \text{ for } \ell \geq 1 \]
Note that $\abs{T_{\ell} \setminus S}$ strictly decrease in each iteration, so we get $T_k = S$, and \[\mu(T_0) \leq k^2 \mu(T_1) \leq k^2 (k-1)^2 \mu(T_2)\leq  \cdots \leq (k!)^{2} \mu(S) \]

\end{proof}

	
\section{Sampling Forests/Spanning Trees}
\label{sec:strees}
In this section we prove \cref{thm:forests,thm:constrainedForests}.\footnote{We remark that reliance on \cref{thm:main} in this section is not mandatory and the results of this section would have been possible even without \cref{thm:main}.}

\begin{figure}
\begin{center}\begin{tikzpicture}
\pic at (-4, 0) {sptrees={ B/C, A/D, A/E, E/F, F/G}};
\begin{scope}[shift={(3.55, -0.55)}, rotate=50]
	\draw[line width=1, dashed] (0, 0) ellipse (1.2 and 0.3);	
\end{scope}
\draw[->, line width=1] (1.5, 0) to node[above, yshift=0.3em] {delete} (2.5, 0);
\pic at (0, 0) {sptrees={B/C, A/D, E/F, F/G, A/E, A/G}};
\draw[->, line width=1] (-2.5, 0) to node[above, yshift=0.3em] {add} (-1.5, 0);
\pic at (4, 0) {sptrees={B/C, A/D, E/F, F/G, A/G}};
\begin{scope}[shift={(-3.45, -0.1)}, rotate=-8]
	\draw[line width=1, dashed] (0, 0) ellipse (1 and 0.25);	
\end{scope}
\end{tikzpicture}\end{center}
\begin{center}\begin{tikzpicture}
\pic at (-4, 0) {sptrees={ B/C, A/D, A/E, E/F, F/G}};
\begin{scope}[shift={(3.55, -0.55)}, rotate=50]
	\draw[line width=1, dashed] (0, 0) ellipse (1.2 and 0.3);	
\end{scope}
\draw[->, line width=1] (1.5, 0) to node[above, yshift=0.3em] {delete} (2.5, 0);
\pic at (0, 0) {sptrees={A/B, B/C, A/D, E/F, F/G, A/E}};
\draw[->, line width=1] (-2.5, 0) to node[above, yshift=0.3em] {add} (-1.5, 0);
\pic at (4, 0) {sptrees={A/B, B/C, A/D, E/F, F/G}};
\begin{scope}[shift={(-3.55, 0.25)}, rotate=30]
	\draw[line width=1, dashed] (0, 0) ellipse (0.85 and 0.25);	
\end{scope}
\end{tikzpicture}\end{center}
\begin{center}\begin{tikzpicture}
 \pic at (-4, 0) {sptrees={ B/C, A/D, A/E, E/F, F/G}};
\draw[->, line width=1] (1.5, 0) to node[above, yshift=0.3em] {no delete} (2.5, 0);
\pic at (0, 0) {sptrees={A/B, B/C, A/D, E/F, F/G, A/E}};
\draw[->, line width=1] (-2.5, 0) to node[above, yshift=0.3em] {add} (-1.5, 0);
\pic at (4, 0) {sptrees={A/B, B/C, A/D, E/F, F/G, A/E}};
\begin{scope}[shift={(-3.55, 0.25)}, rotate=30]
	\draw[line width=1, dashed] (0, 0) ellipse (0.85 and 0.25);	
\end{scope}
\end{tikzpicture}\end{center}
\begin{center}\begin{tikzpicture}
 \pic at (-4, 0) {sptrees={ B/C, A/D, A/E, E/F, F/G}};
\begin{scope}[shift={(3.55, -0.55)}, rotate=50]
	\draw[line width=1, dashed] (0, 0) ellipse (1.2 and 0.3);	
\end{scope}
\draw[->, line width=1] (1.5, 0) to node[above, yshift=0.3em] {delete} (2.5, 0);
\pic at (0, 0) {sptrees={ B/C, A/D, E/F, F/G, A/E}};
\draw[->, line width=1] (-2.5, 0) to node[above, yshift=0.3em] {no add} (-1.5, 0);
\pic at (4, 0) {sptrees={ B/C, A/D, E/F, F/G}};
\end{tikzpicture}\end{center}
\caption{The effect of one step in the down-up random walk on $\mu^{\uparrow}$ on $T_Z = Z \setminus T.$ Four possible ways for $T_Z$ to change. 
}
\label{fig:up-down}
\end{figure}
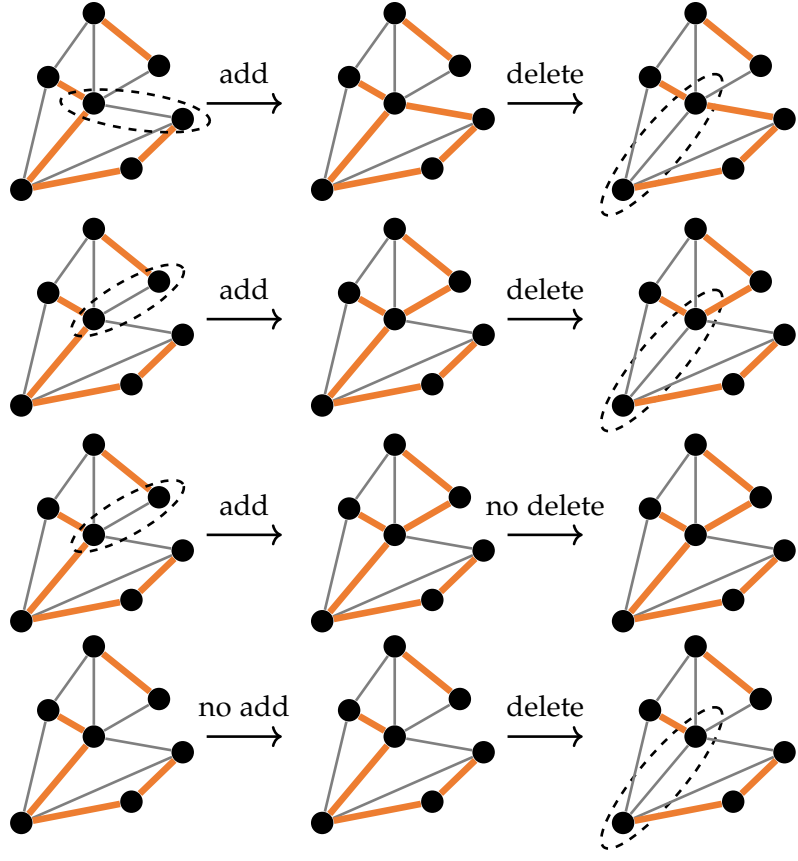

Let $\mu$ be the distribution over forests of $G$ defined in \cref{thm:forests}.
In \cref{lem:complementIndepPoly}, we show a homogeneous multiaffine log-concave polynomial $f_{\matroid,q, w}^{\uparrow}$ that generates homogeneous distribution $\mu^{\uparrow}: \binom{E\cup Y}{n} \to \R_{\geq 0}$
whose projection to $E$ is the \textit{complement} distribution of $\mu$ i.e. $\PrX{T \sim \mu^{\uparrow}}{ T \cap E} = \mu (E \setminus (T \cap E)).$ 

We then run the down-up random walk on the distribution $\mu^{\uparrow}$ for some $t$ steps, and obtain a random set $T_t \in \supp(\mu^{\uparrow}).$ We argue that the distribution of $E \setminus (T_t \cap E)$ is $\epsilon$-close to $\mu$ for some $t= O(\abs{E} \log \frac{\abs{E}}{\epsilon})$ using the mixing time bound proved in \cref{thm:main}. For completeness, we briefly discuss how to implement each step of the random walk in $O(\log \abs{E})$ time.

\begin{lemma} \label{lem:complementIndepPoly}
Let $\matroid$ be a matroid of rank $r$ over ground set $[n]$, and $I(\matroid)$ be its family of independent sets. The following polynomials are completely log-concave.
\begin{enumerate}[label = (\alph*)]
    \item\label{item:LC1} 
\[f_{\matroid}(z_0, z_1, \cdots, z_n) = \sum_{S\in I(\matroid)} z_0^{\abs{S}} z^{[n] \setminus S}\]
\item \label{item:LC2} For 
\[f_{\matroid, q, w} (z_0, z_1, \cdots, z_n)  =  \sum_{S\in I(\matroid)} q^{r - \abs{S}} w^S z_0^{\abs{S}} z^{[n] \setminus S}\]
where $q \geq 0$ and $w_1, \cdots, w_n> 0$
\item\label{item:LC3} 
\[f_{\matroid, q, w}^{\uparrow} (y_1, \cdots, y_r, z_1, \cdots, z_n) = \sum_{S\in I(\matroid), T\in \binom{[r]}{\abs{S}}} \frac{1}{\binom{r}{\abs{S}}} q^{r - \abs{S}} w^S  y^T z^{[n] \setminus S}\]
where $q \geq 0$ and $w_1, \cdots, w_n> 0$
\end{enumerate}

\end{lemma}

\begin{proof}
Since $f_{\matroid,q, w}^{\uparrow} (y_1, \cdots, y_r, z_1, \cdots, z_n) =  \prod^{\uparrow}_{\kappa} (f_{\matroid, q, w})$ with $\kappa_0 = r$ and $\kappa_i = 1\forall i\in [n]$, \cref{item:LC2} implies \cref{item:LC3} by \cref{prop:polarization}.

We show that \cref{item:LC1} implies \cref{item:LC2}. 

If $q > 0$ then
\[f_{\matroid, q, w}(z_0, z_1, \cdots, z_n)    \propto f_{\matroid}(\frac{z_0}{q}, \frac{z_1}{w_1}, \cdots, \frac{z_n}{w_n}) \]
is completely log-concave since composing with a linear map preserves complete log-concavity.

If $q=0$ then
\[f_{\matroid,q, w}(z_0, z_1, \cdots, z_n)  = \sum_{S\in B(\matroid)} w^S  z^{[n]\setminus S} \propto \frac{\partial^{r} f_{\matroid} }{\partial z_0^{r} } (\frac{z_1}{w_1}, \cdots, \frac{z_n}{w_n}) \]
is completely log-concave since taking derivative preserves complete log-concavity.  


Now we show \cref{item:LC1}. Let $f: = f_{\matroid}.$
We first show that $\supp(f)$ is $M$-convex (see Definition \cref{def:m-convex}).
    
    First, by \cref{thm:indepPoly,thm:lorentzian}, the support of $g(y, z_1, \cdots, z_n) = \sum_{S\in \mathcal{I}(\matroid)} y^{n-\abs{S}} w^{S}$ is $M$-convex. Note that 
    \[\text{supp}(f) = \set*{\vec{v}- \tilde{w} \given \tilde{w} \in \text{supp}(g)}\]
    where $v_0 = n$ and $v_i = 1 \forall i \in [n].$ Thus, $\text{supp}(f)$ is also $M$-convex.
    
    Indeed, consider any 
    $\alpha, \beta \in \text{supp}(f)$ and $i \in \{0,\cdots, n\}$ s.t. $\alpha_i < \beta_i$. Then $\vec{v}-\alpha, \vec{v}-\beta \in \supp(g)$ and $(\vec{v}-\alpha)_i > (\vec{v}-\beta)_i $ so there exists $j$ s.t. $(\vec{v}-\alpha)_j < (\vec{v}-\beta)_j$, and $(\vec{v}-\alpha) - e_i +e_j$ and $(\vec{v}-\beta) - e_j + e_i$ are in $\supp(g)$. This implies $\alpha_j > \beta_j$ and $\alpha-e_j +e_i$ and $\beta - e_i+e_j$ are in $\supp(f).$

 We proceed using induction on $n.$ Obviously, for $n=1$, $f$ is a linear function in $w_0, w_1$ with positive coefficient, so is real-stable/Lorentzian. Suppose the statement is true for matroids $\matroid'$ on ground set $[n-1]$ with $n\geq 2.$ 
 
 We only need to verify $\partial^{\alpha} f
 $ is Lorentzian/real-stable for all $\alpha$ with $\abs{\alpha} = n-2.$ Note that for $i\in [n]$, 
    \[\partial_i f   = \sum_{S \in I(\matroid): i \not\in S} z_0^{\abs{S}} z^{([n] \setminus i) \setminus S } = f_{\matroid\setminus i} \]
    is Lorentzian by applying induction hypothesis to $\matroid \setminus i.$ We only need to show $\partial_0^{n-2} f$ is real-stable. Note that $\partial_0^{n-2} f \neq 0$ only if $r =\rank(\matroid) \geq n-2.$ Also, for $n =2$, $\partial_0^{n-2}f$ is exactly $f.$
    
    For $r = n-2,$ $\partial_0^{n-2} f = (n-2)! \sum_{S \in \B(\matroid)} z^{[n] \setminus S} = (n-2)! \sum_{S' \in \B(\matroid^D)} z^{S'} $ is Lorentzian (and of degree 2, thus real-stable), since it is the sum over the bases of dual matroid $\matroid^*$ of $\matroid$ (\cref{thm:basePoly}). 
    
    For $r= n,$ then $\B(\matroid) =\{[n]\}$ and $f = \prod_{i\in [n]} (z_0+z_i)$ is real-stable as product of real stable polynomials $z_0+z_i$ for $i\in [n]$, thus so is $\partial_0^{n-2} f.$ 
    
    For $r=n-1,$ \[\partial_0^{n-2} f = (n-1)! \, z_0 \, e_1(z_T) + (n-2)! \left(e_2(z_T)+ e_1(z_T) e_1(z_{[n] \setminus T})\right )  \] where $T\subseteq [n]$ is such that $\B(\matroid) = \set*{ [n] \setminus i \given i\in T}$ are the bases of $\matroid$, $e_k$ is elementary symmetric polynomial of degree $k$, and $z_S$ is shorthand for $\set*{z_i\given i \in S}$. This is because $\set*{S \in \mathcal{I}(\matroid) \given \abs{S} = n-2}$ is exactly $\set*{[n] - t-t' \given t, t'\in T}\cup \set*{[n] - t- \bar{t} \given t \in T,\bar{t} \not\in T }.$ Set $u := (n-1) z_0 + e_1(z_{[n] \setminus T})$ then
    \[\partial_0^{n-2} f= (n-2)!\left( e_1(z_T) u + e_2(z_T)\right) = (n-2)!\, e_2(u, z_T) \]
    Note that $e_2$ is real-stable (\cref{thm:symmetric}), and that $u\in \mathcal{H}$ whenever $(z_0, z_{[n] \setminus T}) \in \mathcal{H}^{n+1 - \abs{T}},$ where $\mathcal{H}$ is the upper half plane. Thus $\partial_0^{n-2} f$ is nonzero for any $(z_0, z_{[n]}) \in \mathcal{H}^{n+1}$ i.e. $\partial_0^{n-2} f$ is real stable.
\end{proof}
\begin{remark}
Observe that $f_{\matroid}(z_0, z_1, \cdots, z_n) = \sum_{S\in I(\matroid)} z_0^{\abs{S}} z^{[n] \setminus S}$ is the \textit{dual} or \textit{complement} of Lorentzian polynomial $g_{\matroid}(z_0, z_1, \cdots, z_n) = \sum_{S\in I(\matroid)} z_0^{r-\abs{S}} z^{S}$ (see \cref{thm:indepPoly}), in the sense that
\[f_{\matroid}(z_0, z_1, \cdots, z_n) = z_0^r z^{[n]} g_{\matroid}(z_0^{-1}, z_1^{-1}, \cdots, z_n^{-1})\]
The dual of a real-stable polynomial is real-stable, but the dual of a Lorentzian polynomial is not necessarily Lorentzian.
\end{remark}
We are ready to prove \cref{thm:forests}.
\begin{proof}[Proof of \cref{thm:forests}]
Let $\matroid$ be the graphic matroid on graph $G= G(V,E)$ on $n$ edges. Let $k: = \rank(\matroid).$ W.l.o.g, we can label the edges by $1,2\cdots, n$ and assume $E=[n].$ 
Let $\mu$ be the distribution over independent sets of $\matroid$ (i.e. forest of $G$) where $\mu(F) \propto q^{k-\abs{F}} w^F$ for $F \in \mathcal{I}(\matroid)$. Note that we can remove all edge of weight $0$ from $E$ without changing $\mu$. W.l.o.g. we assume this is already done, thus $w(e) > 0 \forall e\in E.$


Let $Y : = \set*{y_1, \cdots, y_k}, Z := \set*{z_1, \cdots, z_n}.$ We identify the variable $z_i$ with the edge labeled by $i$. Let $f: = f^{\uparrow}_{\matroid, q, w}$ then $f$ is multi-affine, homogeneous, and completely log-concave by \cref{lem:complementIndepPoly}. 
Observe that $f$ is the generating polynomial for distribution $\mu^{\uparrow}: \binom{Y \cup Z}{n} \to \R_{\geq 0}$ defined by $\mu^{\uparrow}(T) \propto \frac{q^{k- \abs{T \cap Y}}}{\binom{k}{\abs{T \cap Y}} } w^{Z \setminus T}$ if $Z \setminus T \in \mathcal{I}(\matroid),$ and $0$ otherwise.  

We run the down-up walk starting from $T_0 \in \supp(\mu^{\uparrow})$ (for e.g. $T_0 = (Z\setminus \mathcal{F}) \cup Y$ for some spanning forest $\mathcal{F}$). 
Let $\nu^{\uparrow}$ be the distribution of the set $T \in \sup(\mu^{\uparrow})$ we obtained after $O(n\log(n/\epsilon))$ steps; \cref{thm:main} implies $\norm{\mu^{\uparrow}- \nu^{\uparrow}}_{\tv} \leq \epsilon.$ We then remove all $y_j$ from $S$ i.e. 
collapse $T \subseteq Y \cup Z $ to $T_Z: = Z \setminus (T\cap Z)$. 
Let $\nu$ be the distribution of $T_Z.$ Clearly, $\supp(\nu) = \mathcal{I}(\matroid),$ and if $\nu^{\uparrow}$ is the same distribution as $\mu^{\uparrow}$, then $\nu$ is the same as $\mu.$  
By the data processing inequality, the total variation between $\nu$ and $\mu$ is at most $\epsilon$ since $\norm {\mu - \nu} \leq \norm{\mu^{\uparrow}- \nu^{\uparrow}}_{\tv} \leq \epsilon.$  

We show each step of the random walk can be implemented in $O(\log n)$ time. 
In the down step, we keep track of whether the level $(n-1)$ set $S$ still satisfies $S_Z : = Z \setminus S$ is a forest. Note that if we dropped an $y_{j}$ in the down step to arrive at $S$, then $S_Z$ is always a forest; if, instead, we dropped a $z_i$ (equivalently, added $z_i$ to $S_Z$), then we can check whether $S_Z$ stays a forest in $O(\log n)$ amortized time using link-cut tree \cite{ST83,RTF18}. If $S_Z$ is not a forest, then let $\text{cycle}_S$ be the unique cycle in $S_Z$ which contains the edge $z_i$ that was added to $S_Z$. When we perform an up-step from $S$, if $S_Z$ is not a forest, select $z_f$ among the edges in $\text{cycle}_S$ with probability $\propto 1/w_f$ and add it to $S$ (equivalently, remove $z_f$ from $S_Z$). This can be done in $O(\log n)$-amortized time (see \cite{RTF18}). If $S_Z$ is empty, then we can only add $y_j$ which is not already in $S_Y:=S\cap Y$ with uniform probability. If $S_Z$ is a nonempty forest, then we can add any variable $y_{j} \not\in S_Y$ or $z_i \in S_Z$. In this case, the probability of adding variable can be explicitly computed i.e. uniform among $y_{j}\not\in S_Y$, and $\Pr{z_i} /\Pr{y_{j} } = \frac{q^{k-(\ell-1)}w_i^{-1}/\binom{k}{\ell-1}} {q^{k-\ell}/\binom{k}{\ell}} $ where $\ell: = \abs{S_Z} =\abs{S_Y} +1.$ We can perform these operations in $O(\log n)$ time by 
\begin{itemize}
    \item With probability $1/(1+\tau)$ where $\tau:=\frac{ q^{k-(\ell-1)} \sum_{i \in S_Z} w_i^{-1}/\binom{k} {\ell-1}} {q^{k-\ell} (k-\abs{S_Y}) /\binom{k}{\ell}} = \frac{q\sum_{i \in S_Z}  w_i^{-1}}{\ell}$, sample $y_j$ uniformly at random from $Y \setminus S_Y$ and add $y_j$ to $S_Y.$ Note that this action will always be performed if $q=0$.
    \item Maintain an array of cumulative sums $s_t := \sum_{h=1}^t w_{i_h}^{-1}$ for $t\in [\ell]$ where $w_{i_1} , w_{i_2} , \cdots , w_{i_{\ell}}$ are the weights corresponding to the edges in $S_Z$; this data structure supports amortized $O(\log n)$-time insertion and deletion from $S_Z$ and binary search in the sorted array $[s_t]_{t=1}^{\ell}$. This data structure can be implemented using a splay tree where each node stores the sum of all leaves in its rooted subtree. With probability $\tau/(1+\tau)$, sample $z_f$ from $S_Z$ with probability $\propto 1/w_f$ by: sample uniformly random $p\in [0,s_{\ell}]$, find minimum $t\in [\ell]$ where $p \leq s_t$, and remove $z_{i_t}$ from $S_Z.$ This removal will split a tree in the forest $S_Z$, and we update the link-cut tree representation of $S_Z$ accordingly in $O(\log n)$ time.  
    
    
\end{itemize}
See \cref{fig:up-down} for a visualization of how one up-then-down step may change the set $T_Z = Z \setminus T.$ 

For completeness, we briefly summarize how to handle sampling and removing an edge from $\text{cycle}_S$, which was described in \cite{RTF18}. We represent $S_Z$ as a forest of link-cut trees. When we add an edge $e = (u,v)$ that forms a cycle, splay $u$ to be the root of its tree $\mathcal{T}_u$, then access $v$ (which is also in $T_u$) so that the entire path $\mathcal{P}_{u,v}$ from $u$ to $v$ in $T_u$ is stored in one auxiliary tree. This auxiliary splay tree can be augmented to support (weighted) sampling an edge $f$ from $\mathcal{P}_{u,v}$ as described above with $S_Z.$ Remove $f$ (which disconnects $\mathcal{T}_u$ into two trees) then add $e$. Link-cut trees support these operations in amortized $O(\log n)$ time, and the augmentation increases the run-time by only a constant factor.
\end{proof}
We briefly discuss how to sample from the family of independent sets of a matroid $\matroid = ([n], \mathcal{I})$ of rank $k$ using the framework developed here.  
If we are given access to an oracle $\mathcal{O}'$ whose input-output behavior is as follows:
\begin{itemize}
    \item $\mathcal{O}'$ takes as input a set $S\subseteq [n]$ that is guaranteed to contain at most one circuit
    \item outputs a uniformly random element from the unique circuit in $S$ if it exists
\end{itemize}
then each step of the down-up walk described in Proof of \cref{thm:forests} can be implemented with one call to $\mathcal{O}',$ resulting in a $O(n \log n)$-time algorithm to sample uniformly from the family of independent sets of $\matroid.$ 
We remark that \cref{thm:indepPoly} and the polarization trick employed in Proof of \cref{thm:forests} already give a $O(k n \log\frac{k}{\epsilon})$-time algorithm, given access to the independent set oracle $\mathcal{O}_I$ for $\matroid.$ Indeed, the down-up walk on the distribution defined by the polarization of Lorentzian polynomial $g_{\matroid}(z_1, \cdots, z_n) = \sum_{S \in I(\matroid)} z_0^{k-\abs{S}} z^S $ (see \cref{thm:indepPoly}) mixes in $O(k \log\frac{k}{\epsilon})$ steps, and each step can be implemented using $O(n)$ calls to $\mathcal{O}_{I}$.

We sketch the proof of \cref{thm:constrainedForests}. W.l.o.g., we may assume $0\leq k_1\leq k_2 \leq \rank(\matroid)$ where $\matroid$ is the graphic matroid of $G$.
For any $\ell \leq \rank(\matroid) $, size-$\ell$ forests of $G$ form the basis of a matroid $\matroid^{(\ell)}.$ 
For an arbitrary matroid $\matroid$ and parameters $k_1 \in \N$, we show that $f_{\matroid}^{k_1} (z_0, z_1, \cdots, z_n)= \sum_{S \in \mathcal{I}(\matroid), \abs{S} \geq k_1} z_0^{\abs{S}} z^{[n]\setminus S}  $ is completely log-concave, and proceed as in \cref{lem:complementIndepPoly} and \cref{thm:forests}, while setting $\matroid$ to be $\matroid^{(k_2)}.$ 
\begin{lemma}
For any matroid $\matroid$ of rank $r$ over ground set $[n]$ and parameter $h \in \N$, the following polynomial is completely log-concave
\[f_{\matroid}^{h} (z_0, z_1, \cdots, z_n)= \sum_{S \in \mathcal{I}(\matroid), \abs{S} \geq h} z_0^{\abs{S}} z^{[n]\setminus S}   \]
\end{lemma}
\begin{proof}
Note that if $h> r$ then $f^h_{\matroid}\equiv 0$, and if $h=r$ then $f^h_{\matroid}$ is completely log-concave by apply \cref{thm:basePoly} for the dual matroid of $\matroid$. Below, assume $h < r.$

\cref{lem:complementIndepPoly} implies $f_{\matroid}^0$ is completely log-concave. Now,  $f_{\matroid}^{h} $ and $z_0^h \partial^h_0 f_{\matroid}^0 $ has the same support, and this support is $M$-convex since $z_0^h \partial^h_0 f_{\matroid}^0 $ is Lorentzian (see \cref{lem:clc}, \cref{prop:product}).  

We show $f_{\matroid}^{h} $ is Lorentzian by inducting on $n$ as in \cref{lem:complementIndepPoly}. We only need to check $\partial_0^{n-2} f_{\matroid}^h $ is Lorentzian. 
If $h\leq n-2$, then $\partial_0^{n-2} f_{\matroid}^h$ is exactly $\partial_0^{n-2} f_{\matroid}^0$, thus is Lorentzian because $f_{\matroid}^0$ is Lorentzian. The only remaining case is $h=n-1$ and $r=n$, then $\partial_0^{n-2} f_{\matroid}^h = (n-1)! z_0 \sum_{i=1}^n z_i + \frac{n!}{2} z_0^2 \propto z_0( \sum_{i=1}^n z_i + \frac{n}{2} z_0)$ is real-stable.
\end{proof}



	\printbibliography

@inproceedings{ALOV19,
    author = "Anari, Nima and Liu, Kuikui and Oveis Gharan, Shayan and Vinzant, Cynthia",
    title = "Log-Concave Polynomials II: High-Dimensional Walks and an FPRAS for Counting Bases of a Matroid",
    booktitle = "Proceedings of the 51st Annual {ACM} {SIGACT} Symposium on Theory of Computing",
    publisher = "{ACM}",
    year = "2019",
    month = jun
}

@article{Anari2020IsotropyAL,
  title={Isotropy and Log-Concave Polynomials: Accelerated Sampling and High-Precision Counting of Matroid Bases},
  author={Nima Anari and Michal Derezinski},
  journal={ArXiv},
  year={2020},
  volume={abs/2004.09079}
}

@article{JVV86,
author = {Jerrum, M R and Valiant, L G and Vazirani, V V},
title = {Random Generation of Combinatorial Structures from a Uniform},
year = {1986},
issue_date = {July 1986},
publisher = {Elsevier Science Publishers Ltd.},
address = {GBR},
volume = {43},
number = {2–3},
issn = {0304-3975},
journal = {Theor. Comput. Sci.},
month = jul,
pages = {169–188},
numpages = {20}
}

@MISC{Goel_connectivityin,
    author = {Ashish Goel and Sanjeev Khanna and Sharath Raghvendra and Hongyang Zhang},
    title = {Connectivity in Random Forests and Credit Networks},
    year = {}
}

@book{Fed13,
  title={Theory of optimal experiments},
  author={Fedorov, Valerii Vadimovich},
  year={2013},
  publisher={Elsevier}
}

@article{CGM19,
  title={Modified log-Sobolev inequalities for strongly log-concave distributions},
  author={Cryan, Mary and Guo, Heng and Mousa, Giorgos},
  journal={arXiv preprint arXiv:1903.06081},
  year={2019}
}

@inproceedings{Sch18,
  title={An almost-linear time algorithm for uniform random spanning tree generation},
  author={Schild, Aaron},
  booktitle={Proceedings of the 50th Annual ACM SIGACT Symposium on Theory of Computing},
  pages={214--227},
  year={2018},
  organization={ACM}
}

@book{Oxl06,
  title={Matroid theory},
  author={Oxley, James G},
  volume={3},
  year={2006},
  publisher={Oxford University Press, USA}
}

@article{MV89,
author={M. Mihail and U. Vazirani}, 
title={On the expansion of 0/1 polytopes}, 
journal={Journal of Combinatorial Theory}, 
series={B},
year={1989},
}

@article{ALOV18,
    author = "Anari, Nima and Liu, Kuikui and Oveis Gharan, Shayan and Vinzant, Cynthia",
    title = "Log-Concave Polynomials III: Mason's Ultra-Log-Concavity Conjecture for Independent Sets of Matroids",
    journal = "CoRR",
    volume = "abs/1811.01600",
    year = "2018"
}

@article{KD16,
  title={On sampling and greedy map inference of constrained determinantal point processes},
  author={Kathuria, Tarun and Deshpande, Amit},
  journal={arXiv preprint arXiv:1607.01551},
  year={2016}
}

@article{CM13,
  author    = {Ali {\c{C}}ivril and
               Malik Magdon{-}Ismail},
  title     = {Exponential Inapproximability of Selecting a Maximum Volume Sub-matrix},
  journal   = {CoRR},
  volume    = {abs/1006.4349},
  year      = {2010},
  url       = {http://arxiv.org/abs/1006.4349},
  archivePrefix = {arXiv},
  eprint    = {1006.4349},
  bibsource = {dblp computer science bibliography, https://dblp.org}
}

@article{PackerO4,
  title={Polynomial-time approximation of largest simplices in V-polytopes},
  author={Asa Packer},
  journal={Discret. Appl. Math.},
  year={2004},
  volume={134},
  pages={213-237}
}

@inproceedings{DEFM14,
  title={On largest volume simplices and sub-determinants},
  author={Di Summa, Marco and Eisenbrand, Friedrich and Faenza, Yuri and Moldenhauer, Carsten},
  booktitle={Proceedings of the twenty-sixth annual ACM-SIAM symposium on Discrete algorithms},
  pages={315--323},
  year={2014},
  organization={SIAM}
}

@inproceedings{Nik15,
  title={Randomized rounding for the largest simplex problem},
  author={Nikolov, Aleksandar},
  booktitle={Proceedings of the forty-seventh annual ACM symposium on Theory of computing},
  pages={861--870},
  year={2015}
}

@article{BH18,
  title={Hodge-Riemann relations for Potts model partition functions},
  author={Br{\"a}nd{\'e}n, Petter and Huh, June},
  journal={arXiv preprint arXiv:1811.01696},
  year={2018}
}

@article{BH19,
  title={Lorentzian polynomials},
  author={Br{\"a}nd{\'e}n, Petter and Huh, June},
  journal={arXiv preprint arXiv:1902.03719},
  year={2019}
}

@article{Kha95,
  title={On the complexity of approximating extremal determinants in matrices},
  author={Khachiyan, Leonid},
  journal={Journal of Complexity},
  volume={11},
  number={1},
  pages={138--153},
  year={1995},
  publisher={Elsevier}
}

@misc{wagner2009multivariate,
      title={Multivariate stable polynomials: theory and applications}, 
      author={David G. Wagner},
      year={2009},
      eprint={0911.3569},
      archivePrefix={arXiv},
      primaryClass={math.CV}
}

@inproceedings{AOV18,
    author = "Anari, Nima and Oveis Gharan, Shayan and Vinzant, Cynthia",
    title = "Log-Concave Polynomials I: Entropy and a Deterministic Approximation Algorithm for Counting Bases of Matroids",
    booktitle = "Proceedings of the 59th {IEEE} Annual Symposium on Foundations of Computer Science",
    year = "2018",
    month = oct,
    publisher = "{IEEE} Computer Society",
    doi = "10.1109/focs.2018.00013"
}

@article{BBL09,
  title={Negative dependence and the geometry of polynomials},
  author={Borcea, Julius and Br{\"a}nd{\'e}n, Petter and Liggett, Thomas},
  journal={Journal of the American Mathematical Society},
  volume={22},
  number={2},
  pages={521--567},
  year={2009}
}

@article{KT12,
  title={Determinantal point processes for machine learning},
  author={Kulesza, Alex and Taskar, Ben},
  journal={Foundations and Trends in Machine Learning},
  volume={5},
  number={2--3},
  pages={123--286},
  year={2012},
  publisher={Now Publishers, Inc.}
}

@inproceedings{FM92,
  title={Balanced matroids},
  author={Feder, Tom{\'a}s and Mihail, Milena},
  booktitle={Proceedings of the twenty-fourth annual ACM symposium on Theory of computing},
  pages={26--38},
  year={1992}
}

@inproceedings{Gam99,
  title={On approximating the number of bases of exchange preserving matroids},
  author={Gambin, Anna},
  booktitle={International Symposium on Mathematical Foundations of Computer Science},
  pages={332--342},
  year={1999},
  organization={Springer}
}

@inproceedings{JS02,
  title={Spectral gap and log-Sobolev constant for balanced matroids},
  author={Jerrum, Mark and Son, Jung-Bae},
  booktitle={The 43rd Annual IEEE Symposium on Foundations of Computer Science, 2002. Proceedings.},
  pages={721--729},
  year={2002},
  organization={IEEE}
}

@article{JSTV04,
  title={Elementary bounds on Poincar{\'e} and log-Sobolev constants for decomposable Markov chains},
  author={Jerrum, Mark and Son, Jung-Bae and Tetali, Prasad and Vigoda, Eric},
  journal={The Annals of Applied Probability},
  volume={14},
  number={4},
  pages={1741--1765},
  year={2004},
  publisher={Institute of Mathematical Statistics}
}

@article{Jer06,
  title={Two remarks concerning balanced matroids},
  author={Jerrum, Mark},
  journal={Combinatorica},
  volume={26},
  number={6},
  pages={733--742},
  year={2006},
  publisher={Springer}
}

@article{Clo10,
  title={Approximating the number of bases for almost all matroids},
  author={Cloteaux, Brian},
  journal={Congressus Numerantium},
  volume={202},
  pages={149--153},
  year={2010}
}

@article{CTY15,
  title={Lattice path matroids: negative correlation and fast mixing},
  author={Cohen, Emma and Tetali, Prasad and Yeliussizov, Damir},
  journal={arXiv preprint arXiv:1505.06710},
  year={2015}
}

@article{GJ18,
  title={Approximately counting bases of bicircular matroids},
  author={Guo, Heng and Jerrum, Mark},
  journal={arXiv preprint arXiv:1808.09548},
  year={2018}
}

@inproceedings{DK17,
  title={High dimensional expanders imply agreement expanders},
  author={Dinur, Irit and Kaufman, Tali},
  booktitle={2017 IEEE 58th Annual Symposium on Foundations of Computer Science (FOCS)},
  pages={974--985},
  year={2017},
  organization={IEEE}
}

@article{KM16,
  title={High dimensional random walks and colorful expansion},
  author={Kaufman, Tali and Mass, David},
  journal={arXiv preprint arXiv:1604.02947},
  year={2016}
}

@article{KO20,
  title={High order random walks: Beyond spectral gap},
  author={Kaufman, Tali and Oppenheim, Izhar},
  journal={Combinatorica},
  pages={1--37},
  year={2020},
  publisher={Springer}
}

@article{OR18,
  title={A Polynomial Time MCMC Method for Sampling from Continuous DPPs},
  author={Oveis Gharan, Shayan and Rezaei, Alireza},
  journal={arXiv preprint arXiv:1810.08867},
  year={2018}
}

@article{Aldous90,
  title={The random walk construction of uniform spanning trees and uniform labelled trees},
  author={Aldous, David J},
  journal={SIAM Journal on Discrete Mathematics},
  volume={3},
  number={4},
  pages={450--465},
  year={1990},
  publisher={SIAM}
}

@inproceedings{Broder89,
  title={Generating random spanning trees},
  author={Broder, Andrei Z},
  booktitle={FOCS},
  volume={89},
  pages={442--447},
  year={1989},
  organization={Citeseer}
}

@inproceedings{Wil96,
  title={Generating random spanning trees more quickly than the cover time},
  author={Wilson, David Bruce},
  booktitle={Proceedings of the twenty-eighth annual ACM symposium on Theory of computing},
  pages={296--303},
  year={1996}
}

@article{CMN96,
  title={Two algorithms for unranking arborescences},
  author={Colbourn, Charles J and Myrvold, Wendy J and Neufeld, Eugene},
  journal={Journal of Algorithms},
  volume={20},
  number={2},
  pages={268--281},
  year={1996},
  publisher={Elsevier}
}

@inproceedings{KM09,
  title={Faster generation of random spanning trees},
  author={Kelner, Jonathan A and Madry, Aleksander},
  booktitle={2009 50th Annual IEEE Symposium on Foundations of Computer Science},
  pages={13--21},
  year={2009},
  organization={IEEE}
}

@inproceedings{MST14,
  title={Fast generation of random spanning trees and the effective resistance metric},
  author={Madry, Aleksander and Straszak, Damian and Tarnawski, Jakub},
  booktitle={Proceedings of the twenty-sixth annual ACM-SIAM symposium on Discrete algorithms},
  pages={2019--2036},
  year={2014},
  organization={SIAM}
}

@inproceedings{DKPRS17,
  title={Sampling random spanning trees faster than matrix multiplication},
  author={Durfee, David and Kyng, Rasmus and Peebles, John and Rao, Anup B and Sachdeva, Sushant},
  booktitle={Proceedings of the 49th Annual ACM SIGACT Symposium on Theory of Computing},
  pages={730--742},
  year={2017}
}

@inproceedings{DPPR17,
  title={Determinant-preserving sparsification of SDDM matrices with applications to counting and sampling spanning trees},
  author={Durfee, David and Peebles, John and Peng, Richard and Rao, Anup B},
  booktitle={2017 IEEE 58th Annual Symposium on Foundations of Computer Science (FOCS)},
  pages={926--937},
  year={2017},
  organization={IEEE}
}

@article{MT06,
  title={Mathematical aspects of mixing times in Markov chains},
  author={Montenegro, Ravi and Tetali, Prasad},
  journal={Foundations and Trends{\textregistered} in Theoretical Computer Science},
  volume={1},
  number={3},
  pages={237--354},
  year={2006},
  publisher={Now Publishers, Inc.}
}

@article{RTF18,
  title={Linking and Cutting Spanning Trees},
  author={Russo, Luis and Teixeira, Andreia Sofia and Francisco, Alexandre P},
  journal={Algorithms},
  volume={11},
  number={4},
  pages={53},
  year={2018},
  publisher={Multidisciplinary Digital Publishing Institute}
}

@article{MS99,
  title={M-convex function on generalized polymatroid},
  author={Murota, Kazuo and Shioura, Akiyoshi},
  journal={Mathematics of operations research},
  volume={24},
  number={1},
  pages={95--105},
  year={1999},
  publisher={INFORMS}
}

@article{ST83,
  title={A data structure for dynamic trees},
  author={Sleator, Daniel D and Tarjan, Robert Endre},
  journal={Journal of computer and system sciences},
  volume={26},
  number={3},
  pages={362--391},
  year={1983},
  publisher={Elsevier}
}

@book{CT12,
  title={Elements of information theory},
  author={Cover, Thomas M and Thomas, Joy A},
  year={2012},
  publisher={John Wiley \& Sons}
}

@article{Abe80,
  title={On the Pl{\"u}cker relations for the Grassmann varieties},
  author={Abeasis, Silvana},
  journal={Advances in Mathematics},
  volume={36},
  number={3},
  pages={277--282},
  year={1980},
  publisher={Academic Press}
}

@article{MS18,
  title={Simpler exchange axioms for M-concave functions on generalized polymatroids},
  author={Murota, Kazuo and Shioura, Akiyoshi},
  journal={Japan Journal of Industrial and Applied Mathematics},
  volume={35},
  number={1},
  pages={235--259},
  year={2018},
  publisher={Springer}
}

@incollection{Gur09,
  title={On multivariate Newton-like inequalities},
  author={Gurvits, Leonid},
  booktitle={Advances in combinatorial mathematics},
  pages={61--78},
  year={2009},
  publisher={Springer}
}

@article{AHK18,
  title={Hodge theory for combinatorial geometries},
  author={Adiprasito, Karim and Huh, June and Katz, Eric},
  journal={Annals of Mathematics},
  volume={188},
  number={2},
  pages={381--452},
  year={2018},
  publisher={JSTOR}
}

@article{Nui68,
 ISSN = {00255521, 19031807},
 URL = {http://www.jstor.org/stable/24489791},
 author = {WIM NUIJ},
 journal = {Mathematica Scandinavica},
 number = {1},
 pages = {69--72},
 publisher = {Mathematica Scandinavica},
 title = {A NOTE ON HYPERBOLIC POLYNOMIALS},
 volume = {23},
 year = {1969}
}

@article{asner70,
 ISSN = {00361399},
 URL = {http://www.jstor.org/stable/2099475},
 author = {Bernard A. Asner},
 journal = {SIAM Journal on Applied Mathematics},
 number = {2},
 pages = {407--414},
 publisher = {Society for Industrial and Applied Mathematics},
 title = {On the Total Nonnegativity of the Hurwitz Matrix},
 volume = {18},
 year = {1970}
}
\end{document}